\documentclass{amsart}

\usepackage{amsmath,amssymb}
\usepackage[foot]{amsaddr}

\usepackage{changepage}

\usepackage[utf8x]{inputenc}

\usepackage{textcomp,marvosym}

\usepackage{cite}

\usepackage{nameref,hyperref}

\usepackage[right]{lineno}

\usepackage{microtype}
\DisableLigatures[f]{encoding = *, family = * }

\usepackage[table]{xcolor}

\usepackage{multirow,array}

\usepackage{algorithm}
\usepackage[noend]{algorithmic}
\usepackage{pgfplots}
\usepackage{latexsym}
\usepackage{dblfloatfix}
\usepackage[colorinlistoftodos]{todonotes}

\DeclareGraphicsExtensions{.pdf,.jpeg,.png}
\newcolumntype{+}{!{\vrule width 2pt}}

\newlength\savedwidth

\newcommand{\colorset}{\mathcal{C}}
\newcommand{\colorf}{c}

\newtheorem{theorem}{Theorem}
\newtheorem{corollary}{Corollary}
\newtheorem{lemma}{Lemma}

\usetikzlibrary{calc}
\usetikzlibrary{positioning}

\pgfplotsset{ 
    discard if not/.style 2 args={
        x filter/.code={
            \edef\tempa{\thisrow{#1}}
            \edef\tempb{#2}
            \ifx\tempa\tempb
            \else
                
            \fi
        }
    }, 
}

\selectcolormodel{cmyk}

\begin{document}

\title[On motifs in colored graphs]{On motifs in colored graphs$^\star$}


\author[D.~P.~Rubert]{Diego P.~Rubert$^1$}
\thanks{$^\star$ This is a pre-print of an article to be published in the Journal of Combinatorial Optimization.} 

\author[E.~Araujo]{Eloi Araujo$^1$}

\author[M.~A.~Stefanes]{Marco A.~Stefanes$^1$}
\email{\{diego,feloi,marco,fhvm\}@facom.ufms.br}

\author[J.~Stoye]{Jens Stoye$^2$}
\email{jens.stoye@uni-bielefeld.de}

\author[F.~V.~Martinez]{F\'abio V.~Martinez$^{1,\dagger}$}
\thanks{$^\dagger$ Corresponding author}

\address{$^1$ Faculdade de Computa\c c\~ao, Universidade Federal de Mato Grosso do Sul, Brazil}
\address{$^2$ Faculty of Technology and Center for Biotechnology (CeBiTec), Bielefeld University, Germany}




\begin{abstract}
One of the most important concepts in biological network analysis is
that of network motifs, which are patterns of interconnections that
occur in a given network at a frequency higher than expected in a
random network. In this work we are interested in searching and
inferring network motifs in a class of biological networks that can be
represented by vertex-colored graphs. We show the computational
complexity for many problems related to colorful topological motifs
and present efficient algorithms for special cases. We also present a
probabilistic strategy to detect highly frequent motifs in
vertex-colored graphs.  Experiments on real data sets show that our
algorithms are very competitive both in efficiency and in quality of
the solutions.
\end{abstract}

\maketitle

\section{Introduction}

Biological networks have been wide and deeply studied in recent years.
Their analysis provides comprehension of the underlying biological
processes, the function and the structure of their components, and of
their evolutionary relationships.  Such networks may be reaction
graphs, regulatory networks, protein-protein-interaction (PPI)
networks or metabolic pathways, all of which can be modeled by
vertex-colored graphs where the vertices represent biochemical
reactions, genes, proteins or metabolites.

A concept that takes a central role in such network analysis is the so
called \emph{network motif}, i.e., a pattern of interconnections that
occur at higher frequency inside a (biological) network than expected
in a random network~\cite{Wernicke2006}.

In this context, we consider two problem variants: \emph{Motif search}
is a graph-theoretic pattern matching problem where a small graph
(pattern) is searched in a large graph, such that structure of
occurrence of the pattern is preserved, i.e.\ its topology and
connectivity.  \emph{Motif inference} (or \emph{motif discovery})
receives as input only one (large) graph and the task is to detect
(small) subgraphs occuring at high frequency.

The first approaches to motif search in a biological network were
proposed in the context of transcriptional regulation
networks~\cite{SMMA2002} and later within PPI
networks~\cite{Kelley-etal-2003}. Since most variations of the motif
search problem are computationally
hard~\cite{LFS2006,FFHV2011,DFV2011}, several techniques are used in
order to overcome that
hardness~\cite{LFS2006,Bruckner-etal-2010,Blin-etal-2010,Kashani-etal-BMC2009}. Searching
motifs with a specific topology arose in Shlomi \textit{et
  al.}~\cite{SSRS2006} using fixed-parameter tractable algorithms for
searching motifs defined as paths within a PPI network. A more general
solution was implemented in the \textsf{Qnet} tool which searches for
motifs defined as trees~\cite{Shlomi-etal-2008}.

Lacroix \textit{et al.}~\cite{LFS2006} proposed a new approach for the
motif search problem in vertex-colored graphs, where the motif
topology is not taken into consideration and only connectivity is
required. Such a motif is called \emph{colored motif}. This motif
search problem has received much
attention~\cite{GS2013,PZ-JDA-2014,PHM-DAM2016}. A fixed-parameter
algorithm was presented by Lacroix \textit{et al.}~\cite{LFS2006},
extended to infer all colored motifs in metabolic networks and
implemented as a tool called \textsf{MOTUS}~\cite{SLS2009}.
\textsf{Torque}~\cite{Bruckner-etal-2010} is another solution that
aims to search a given colorful motif in a PPI network. (A motif is
\emph{colorful} if each color occurs only once.) Later,
\textsf{GraMoFoNe}~\cite{Blin-etal-2010} generalized the problem,
searching for a given colored motif, not only colorful, in a PPI
network.

Most of the previous implementations are concerned only with motif
search. Another important task is enumerating all motifs of a given
size in a network. \textsf{MOTUS} can solve such a problem only for
small motifs (of size up to 7), due to its computational cost.

This paper is an extension of two previous works presented in
conferences~\cite{AS2013,RAS2015} where we showed the hardness of some
colorful topological motif problems and presented algorithms for
searching and enumerating colorful motifs. Many of those previous
results are, for the sake of clarity, also presented in this paper and
can be summarized as follows: ({\it i}) Given a vertex-colored graph
$G$ representing a biological network and a colorful topological motif
$M$, finding a simple subgraph of $G$ isomorphic to $M$ is NP-hard
(Theorem~\ref{theo:SMS-NP-c}); ({\it ii}) Given a vertex-colored graph
$G$ and a colorful motif $M$, finding an induced subgraph of $G$
isomorphic to $M$ is NP-hard (Theorem~\ref{theo:ISM}); ({\it iii})
Given a vertex-colored graph $G$ and a colorful tree $M$, finding $M$
as a simple subgraph of $G$ can be performed in polynomial time
(Algorithm~\ref{alg:TCG}); ({\it iv}) Given a vertex-colored graph $G$
and a colorful tree $M$, an algorithm to enumerate all occurrences of
$M$ in $G$ is provided (Algorithm~\ref{alg:all-colorful}); ({\it v}) A
framework for evaluating high frequency and inferring motifs in
colored-vertex graphs is provided (Section~\ref{sec:infe}); and
({\it vi}) Experimental evaluation of algorithms for motif search and
inference is performed, including analysis of the high frequency
evaluation framework and comparison of these algorithms with other
similar tools (Section~\ref{sec:expe}).

Besides that, in this paper we improve some of those previous results
and present new results, namely:
\begin{enumerate}
\item[({\it i})] We propose a new related problem where given a
  vertex-colored graph $G$, a colorful motif $M$ and an integer $k$,
  we want to find at least $k$ (vertex) disjoint occurrences of $M$
  in $G$. We prove this is an NP-hard problem (Theorem~\ref{result6});
\item[({\it ii})] We also propose a problem where given two
  vertex-colored graphs $G$ and $H$ and an integer $k$, we want to
  find a colorful tree with $k$ vertices which is a subgraph of both
  $G$ and $H$. We prove this is an NP-hard problem
  (Theorem~\ref{theo:common-k-tree});
\item[({\it iii})] We show, using a particular data structure, that
  the algorithm for finding a subgraph of a given vertex-colored graph
  $G$ isomorphic to a given colorful motif $M$
  (Algorithm~\ref{alg:TCG}) can be improved from quadratic to linear
  running time (Theorem~\ref{theo:TCG-time});
\item[({\it iv})] We present a method (Algorithm~\ref{alg:numocc})
  that, given a vertex-colored graph $G$ and a colorful tree motif
  $M$, computes the number of occurrences of $M$ in $G$, showing that
  it can be performed in linear time (Theorem~\ref{theo:numocc});
\item[({\it v})] We speedup considerably the previous version of our
  motif inference algorithm (Algorithm~\ref{alg:inference}) by using
  as subroutine the new algorithm for counting the number of
  occurrences of motifs (Algorithm~\ref{alg:numocc}), updating
  experimental tests (Section~\ref{sec:expe}) related to this
  algorithm.
\end{enumerate}

This paper is organized as follows. Section~\ref{sec:prel} provides
basic definitions and notations. Section~\ref{sec:comp} presents
complexity results for problems of interest. Search algorithms such as
for finding subgraphs isomorphic to colorful trees, finding a maximum
clean graph, finding all colorful motifs in a clean subgraph, and
finding the number of occurrences of colorful motifs in a graph are
given in Section~\ref{sec:sear}. Next, Section~\ref{sec:infe} provides
a framework to detect subgraphs in a network occurring with high
frequency, and also an algorithm for inferring statistically
significant colorful motifs in a given vertex-colored
graph. Section~\ref{sec:expe} shows experimental results of
implementations of sequential and parallel proposed algorithms. A
conclusion is given in Section~\ref{sec:conc}.

\section{\label{sec:prel} Preliminaries}

Let $G = (V_G, E_G)$ be a graph such that $V_G$ is a set of vertices
and $E_G \subseteq \binom{V_G}{2}$ is a set of edges (unordered pairs
of vertices). Let $\colorset$ be a set of colors. A \emph{color
  function} $\colorf : X \rightarrow \colorset$ assigns a color in
$\colorset$ to each element in $X$. Typically, $X$ is a subset of the
set of vertices of a graph. Thus, a \emph{vertex-colored graph} is a
graph with colored vertices, i.e., a graph $G$ such that $\colorf :
V_G \rightarrow \colorset$ assigns a color in $\colorset$ to each
vertex in $V_G$. Sometimes in this text, we refer to a vertex-colored
graph simply as a graph.

We denote by $uv$ an edge $\{u, v\}$ of a graph and we say that
vertices $u$ and $v$ are \emph{adjacent}. A \emph{(color-) isomorphism}
between graphs $G$ and $H$ is a bijection $f$ from $V_G$ to $V_H$ such
that, for any two elements $u,v \in V_G$, $u$ and $v$ are
adjacent in $G$ if and only if $f(u)$ and $f(v)$ are adjacent in
$H$. Furthermore, for each pair $(v, f(v))$, where $v \in V_G$ and
$f(v) \in V_H$, we have $\colorf(v) = \colorf(f(v))$. Graphs $G$ and
$H$ are \emph{(color-) isomorphic} if there exists an isomorphism
between them, and we denote this by $G \cong H$.

Let $G$ be a graph, define $\colorf(X) := \{\colorf(x) : x \in X\}$
for any subset $X \subseteq V_G$. A \emph{subgraph} of $G$ is a graph
$H$ such that $V_H \subseteq V_G$ and $E_H \subseteq E_G$. We denote
it by $H \subseteq G$. A subgraph of $G$ \emph{induced} by a subset
$X$ of $V_G$ is the graph $H$ such that $V_H = X$ and $E_H = E_G \cap
\binom{X}{2}$.  Such a graph $H$ is also denoted by $G[X]$.

For a subset $X$ of $V_G$ of a graph $G$, we denote by $G - X$ the
subgraph $G[V_G \setminus X]$. Similarly, if $A$ is a subset of $E_G$,
we denote by $G - A$ the graph $(V_G, E_G \setminus A)$. Moreover,
when it is implicit where adding a vertex and an edge in a graph $G$
(for instance, adding a vertex $x$ and an edge $xy$ in a graph $G$,
with $y \in V_G$), then we use the notation $G + x$. In some cases, we
use the notation $G + xy$ to denote the addition of edge $xy$ and the
vertex $x$ to $G$, supposing that $y \in V_G$.  We write $H \prec G$
if $H$ is isomorphic to a subgraph of $G$.

Let $G$ be a vertex-colored graph and $M \prec G$. The \emph{restriction}
of color function $\colorf$ to $V_M$
is the function $\colorf_M = (\colorf | V_M)$
and thus $M$ is also a colored graph. If $\colorf_M$ is a injection we
say that $M$ is a \emph{colorful} (sub)graph (informally, each color
in $\colorset$ is assigned to at most one vertex in $V_M$). 
In this setting, we say that such a graph is a \emph{colorful motif}.

An \emph{exact occurrence} of a colorful motif $M$ in a vertex-colored
graph $G$ is a vertex set $S$ in $V_G$ such that $G[S]$ is connected,
$M \prec G[S]$ and $|V_M| = |S|$. In this paper, we are interested in
the following computational problems.

\medskip

\noindent \textbf{Problem} \textsc{Subgraph-Motif}($G, M$): given a
vertex-colored graph $G$ and a colorful motif $M$, does there exist a
subgraph of $G$ isomorphic to $M$, that is, $M \prec G$?

\medskip

\noindent \textbf{Problem} \textsc{Induced-Subgraph-Motif}($G, M$):
given a ver\-tex-colored graph $G$ and a colorful motif $M$, does
there exist an induced subgraph of $G$ isomorphic to $M$?

\medskip

\noindent \textbf{Problem} \textsc{All-Motifs}($G, M$): given a
vertex-colored graph $G$ and a colorful motif $M$, enumerate all occurrences
of $M$ in $G$, such that $M$ is a subgraph of~$G$.

\medskip

\noindent \textbf{Problem} \textsc{$k$-Disjoint-Motifs}($G, M, k$):
given a vertex-col\-ored graph $G$, a colorful motif $M$, and an
integer $k > 0$, do there exist at least $k$ disjoint occurrences of
$M$ in $G$, such that $M$ is a subgraph of $G$?

\medskip

\noindent \textbf{Problem} \textsc{Common-$k$-Tree}($G, H, k$): given
two vertex-colored graphs $G$ and $H$, and an integer $k > 0$, does
there exist a colorful tree $T$ with $k$ vertices such that $T \prec
G$ and $T \prec H$?

\section{\label{sec:comp} Complexity results}

We start this section by proving the computational complexity of
\textsc{Subgraph-Motif}.

\begin{theorem} \label{theo:SMS-NP-c}
Problem \textsc{Subgraph-Motif} is NP-complete.
\end{theorem}

\begin{proof}
  We first show that \textsc{Subgraph-Motif} belongs to NP. Given a
  graph $G$ and a colorful motif $M$, the certificate is a graph $G'$
  such that $G' \prec G$. A verification algorithm can easily check
  this in polynomial time.

  We provide a reduction from the \textsc{3-sat} problem, which is
  NP-complete~\cite{GJ1979}. Given an arbitrary Boolean formula $\Phi$
  in conjunctive normal form (CNF) with $m$ clauses $C_1, \ldots, C_m$
  as an instance of \textsc{3-sat}, we construct $G$ with $3m$
  vertices, where each vertex represents a literal of a clause in
  $\Phi$. Vertices $u$ and $v$ in $G$ have the same color if and only
  if $u$ and $v$ come from the same clause in $\Phi$. An edge $uv \in
  E_G$ if and only if literals representing vertices $u$ and $v$ are
  not opposite. We construct $M$ as a colorful clique with $m$
  vertices, whose colors are the $m$ distinct colors of vertices in
  $G$ (Fig~\ref{fig:inducedNPcompleteTeo1}). Clearly, this
  transformation can be done in polynomial time.

  \begin{figure}[!h]
    \begin{center}
      \includegraphics{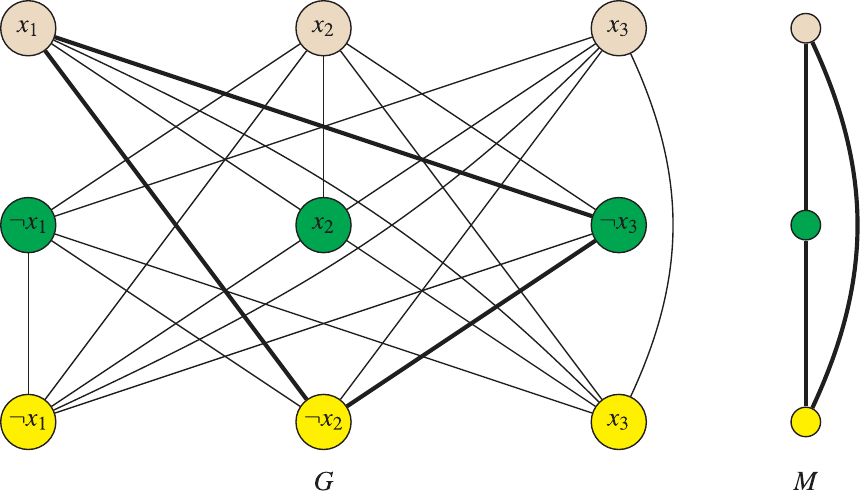}
      \caption{\label{fig:inducedNPcompleteTeo1} Graphs $G$ and $M$
        constructed from formula $\Phi = C_{1} \wedge C_{2} \wedge
        C_{3}$, where $C_{1} = x_{1} \vee x_{2} \vee x_{3}$, $C_{2} =
        \neg x_{1} \vee x_{2} \vee \neg{x}_{3}$, and $C_{3} =
        \neg{x}_{1} \vee \neg x_{2} \vee x_{3}$. Edges whose literals
        come from the same clause are not represented. Notice that
        $x_{1}, \neg x_{2}, \neg x_{3}$ is a truth assignment satisfying
        $\Phi$ and $M \prec G$.}
    \end{center}
  \end{figure}

  We argue that $\Phi$ is satisfiable if and only if $M \prec
  G$. Suppose that $\Phi$ is satisfiable. Then there exists a truth
  assignment to the variables satisfying all clauses. Let $S$ be a set
  of vertices in $G$ corresponding to $m$ literals with values
  ``true'', one for each clause, that satisfies $\Phi$. By definition
  of $G$, it follows that $G[S] \cong M$ and thus $M \prec G$. On the
  other hand, suppose that $M \prec G$. Then there exists a subgraph
  $H \subseteq G$ such that $H \cong M$. It means that $H$ is a
  colorful clique and it follows that the set of vertices in $H$
  represents a set $S$ of $m$ non-opposite literals in $\Phi$, one per
  clause. Consequently, a truth assignment to literals in $S$
  satisfies $\Phi$.
\end{proof}

Now we show the computational complexity for the related problem
\textsc{Induced-Subgraph-Motif}.

\begin{theorem} \label{theo:ISM}
Problem \textsc{Induced-Subgraph-Motif} is NP-complete, even when the given
colorful motif is a tree.
\end{theorem}

\begin{proof}
  As for \textsc{Subgraph-Motif} (Theorem~\ref{theo:SMS-NP-c}),
  we can show easily that \textsc{In\-duced-Subgraph-Motif} belongs to
  NP.

  To show NP-hardness, we present a reduction from \textsc{3-sat},
  similar to Theorem~\ref{theo:SMS-NP-c}. Given an arbitrary Boolean
  formula $\Phi$ in CNF with $m$ clauses $C_1, \ldots C_m$ as an
  instance of \textsc{3-sat}, we construct a graph $G$ with $3m+1$
  vertices, where each one of the first $3m$ vertices represents a
  literal of a clause in $\Phi$. The last extra vertex is called
  \emph{core}. Vertices $u$ and $v$ in $G$ have the same color if and
  only if $u$ and $v$ come from the same clause in $\Phi$. The core
  vertex has a color different from any other vertex, and it is
  adjacent to each one of the $3m$ remaining vertices. Moreover, two
  vertices representing literals are adjacent if and only if the
  literals representing them are opposite. Furthermore, $M$ is a
  colorful star with $m+1$ vertices, whose colors are in $V_G$ and the
  color of its center is the color of the core
  (Fig~\ref{fig:inducedNPcomplete}). Notice that such a
  transformation can be performed in polynomial time.

  \begin{figure}[!h]
    \begin{center}
      \includegraphics{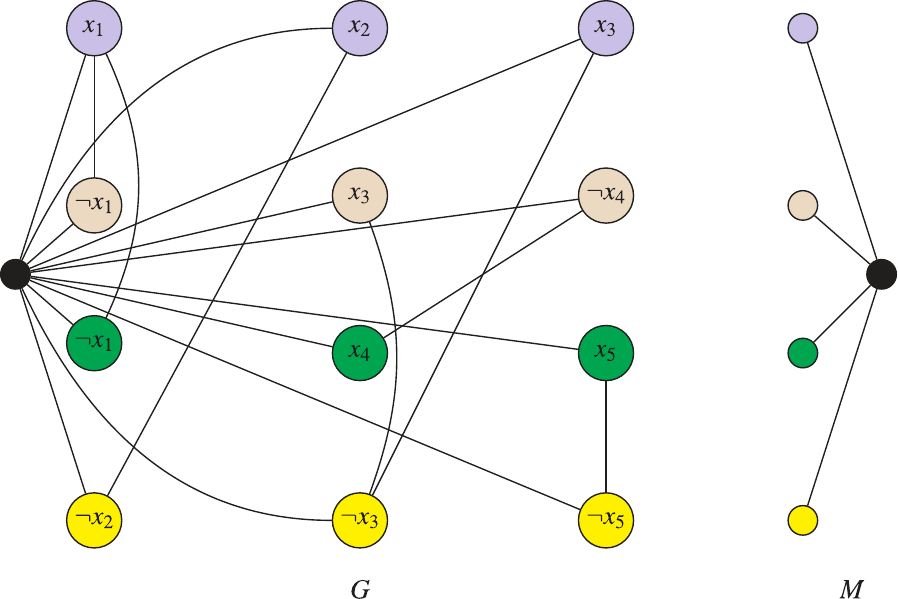}
      \caption{\label{fig:inducedNPcomplete} Graphs $G$ and $M$
        constructed from formula $C_{1} \wedge C_{2} \wedge C_{3} \wedge
        C_{4}$, where $C_{1} = x_{1} \vee x_{2} \vee x_{3}$, $C_{2} =
        \neg x_{1} \vee x_{3} \vee \neg{x}_{4}$, $C_{3} = \neg{x}_{1}
        \vee x_{4} \vee x_{5}$, and $C_{4} = \neg{x}_{2} \vee
        \neg{x}_{3} \vee \neg{x}_{5}$.}
    \end{center}
  \end{figure}

  Now, we show that $\Phi$ is satisfiable if and only if there exists
  an induced subgraph of $G$ isomorphic to star~$M$. Sup\-pose that
  $\Phi$ is satisfiable. Then, there exists a truth as\-sign\-ment to
  the variables satisfying all clauses. Let $S$ be a set of vertices
  corresponding to $m$ literals with values ``true'', belonging to
  different clauses, that satisfy $\Phi$. By definition of $G$, we
  have $G[S \cup \{ \text{core} \}] \cong M$ and thus $G[S \cup \{
    \text{core} \}]$ is an induced subgraph of $G$ isomorphic to
  $M$. Conversely, suppose that there exists an induced subgraph $H$
  of $G$ isomorphic to $M$. Since the color of the core is unique in
  $G$ and $M$, the color of the center of $H$ has the same color as
  the vertex core in $M$ and the set of vertices in $V_H - \{
  \text{core} \}$ represents a set $S$ of $m$ not opposite literals in
  $\Phi$, one per clause. Hence, a truth assignment to literals in $S$
  satisfies $\Phi$.
\end{proof}

\newcommand{\bfit}[1]{\text{\textbf{\textit{#1}}}}

In the following we prove the computational complexity for
\textsc{$k$-Disjoint-Motifs} using a reduction from a well-known
computational problem \textsc{Longest-Common-Subsequence}. Before
doing that we need a few definitions.  Given a finite set of sequences
$R = \{S_1, \ldots, S_p\}$, we denote by $s_{i,j}$ the $j$th symbol in
sequence $S_i$. We say that a $p$-tuple of integers $\bfit{j} = (j_1,
\ldots, j_p)$ is a \emph{column} in $R$ if $s_{i,j_i} =
s_{i+1,j_{i+1}}$ for each $1 \le i < p$. We also say that columns
$\bfit{j} = (j_1,\ldots,j_p)$ and $\bfit{k} = (k_1,\ldots,k_p)$ are
\emph{crossing} in $R$ when there exists some $i$, $1\le i < p$, such
that $j_i \le k_i$ and $j_{i+1} \geq k_{i+1}$ (or $j_i \ge k_i$ and
$j_{i+1} \le k_{i+1}$).  Otherwise, they are \emph{non-crossing},
i.e., when $j_i < k_i$ (or $j_i > k_i$) for all $i$, $1\le i\le
p$. Given a set $C$ of columns in $R$, we say that $C$ is a
\emph{common subsequence in $R$} if $\bfit{j}$ and $\bfit{k}$ are
non-crossing for each pair of distinct columns $\bfit{j}$ and
$\bfit{k}$ in $C$. Fig~\ref{fig:commonsubseq} presents a graphical
example of a common subsequence.

\begin{figure}[h]
  \begin{center}
    \includegraphics{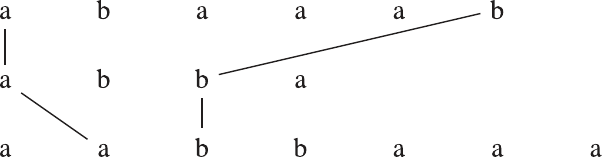}
    \caption{\label{fig:commonsubseq} A graphical representation of a
      common subsequence $C = \{(1, 1, 2), (6, 3, 3)\}$, with two
      non-crossing columns, of three sequences.}
  \end{center} 
\end{figure}


Now we can present the following decision problem, which has been
shown to be NP-complete~\cite{maier1978}:

\medskip

\noindent \textbf{Problem} \textsc{Longest-Common-Subsequence}($R,
k$): given a set $R$ of sequences and an integer $k$, does there exist
a common subsequence $C$ in $R$ such that $|C| \geq k$?

\medskip

We are ready to demonstrate the following result.

\begin{theorem}\label{result6}
Problem \textsc{$k$-Disjoint-Motifs} is NP-complete.
\end{theorem} 

\begin{proof}
  It is easy to show that \textsc{$k$-Disjoint-Motifs} belongs to
  NP. Given a graph $G$, a colorful motif $M$, an integer $k$, and a
  set of $k' \geq k$ subgraphs of $G$, it can be checked, in
  polynomial time, if each of those $k'$ subgraphs is isomorphic to
  $M$ and no two of them share a vertex.

  To complete the NP-hardness proof, we present a reduction from
  \textsc{Longest-Common-Subsequence}. Given a set of sequences $R =
  \{S_1, \ldots, S_p\}$ and an integer $k$ as an input of
  \textsc{Longest-Common-Subsequence}, let us construct a graph $G$
  and a colorful motif $M$ such that there exists a common subsequence
  in $R$ of length at least $k$ if and only if there exist at least
  $k$ disjoint occurrences of $M$ in $G$.

  \noindent (Graph $G$). We describe the graph $G$ constructed from
  $R$. The set of vertices $V_G$ is partitioned into \emph{sequence
    vertices} $V^S$, \emph{core vertices} $V^C$, and \emph{transversal
    vertices} $V^X$. Two vertices in different subsets have different
  colors. The set of edges $E_G$ is partitioned into \emph{core edges}
  $E^C$ and \emph{transversal edges} $E^X$.

  \noindent ((Sequence vertices $V^S$)). Each symbol $s_{i,j}$ in a
  sequence of $R$ is represented by a vertex $v_i^j$ in $V^S$. All
  vertices from the sequence $S_i$ have color $c_i$, and $c_i \not=
  c_{i'}$ for $i \not= i'$.

  \noindent ((Core vertices $V^C$ and core edges $E^C$)). The set of
  core vertices is partitioned into $p - 1$ disjoint sets $V_1^C,
  \ldots, V_{p-1}^C$. Each pair of vertices $(v_i^j, v_{i+1}^{j'})$
  such that $s_{i,j} = s_{i+1,j'}$ is represented by a core vertex
  $\mathbf{v}_i^z$ in $V_i^C$, where $z = (v_i^j, v_{i+1}^{j'})$. All
  vertices in $V_i^C$ have color $d_i$, and $d_i \neq d_{i'}$ for $i
  \not= i'$. The set $E^C$ is partitioned into $p - 1$ disjoint
  sets $E_1^C, \ldots, E_{p-1}^C$, and for each core vertex
  $\mathbf{v}_i^z$, representing a pair $z = (v_i^j, v_{i+1}^{j'})$,
  we have edges $(v_i^j, \mathbf{v}_i^z)$ and $(\mathbf{v}_i^z,
  v_{i+1}^{j'})$ in $E_i^C$. We say that a pair of core vertices
  $(\mathbf{v}_i^y, \mathbf{v}_i^z)$ in $V_i^C$, with $y =
  (v_i^j,v_{i+1}^{j'})$ and $z = (v_i^k, v_{i+1}^{k'})$, is
  \emph{opposing in $i$} if $j < k$ and $j' > k'$.
  
  \noindent ((Transversal vertices $V^X$ and transversal edges
  $E^X$)). Sets $V^X$ and $E^X$ are partitioned into $p-1$ sets
  $V_1^X, \ldots, V_{p-1}^X$ and $E_1^X, \ldots, E_{p-1}^X$,
  respectively. Two vertices in distinct sets $V_i^X$ and $V_{i'}^X$
  have different colors. Let $\chi_i$ be the set of pairs of (core)
  vertices opposing in $i$ and suppose that $V^C_i = \{\mathbf{v}_i^1,
  \ldots, \mathbf{v}_i^\ell\}$ with $\ell = |V_i^C|$. Now, observe
  that we relabel arbitrarily the superscripts of vertices in $V_i^C$
  with integer numbers. For each $x$ in $\chi_i$, we list a sequence
  $x^1, x^2, \ldots, x^\ell$ such that $x^j = x^{j'}$ if the pair of
  core vertices $\mathbf{v}_i^j$ and $\mathbf{v}_i^{j'}$ is opposing
  in $i$. Then, we add all the $\ell - 1$ distinct elements of this
  sequence to the set of transversal vertices $V_i^X$. Each core
  vertex $\mathbf{v}_i^j$ in $V_i^C$ is adjacent to the transversal
  vertex $x^j$ in $V_i^X$. Furthermore, two transversal vertices in
  $V_i^X$ have the same color if and only if they come from a pair of
  core vertices opposing in $i$. This completes the description of
  $G$.

  \noindent (Motif $M$). The colorful tree $M$ consists of a path of
  length $2p-1$, where vertices have colors $c_1, d_1, c_2, d_2,
  \ldots, c_{p-1}, d_{p-1}, c_p$. Moreover, each vertex of color $d_i$
  is adjacent to $|\chi_i|$ vertices whose colors are those used for
  coloring the transversal vertices in $V_i^X$. See
  Fig~\ref{fig:disjointmotifs}.

  \begin{figure}[h]
    \begin{center}
      \includegraphics{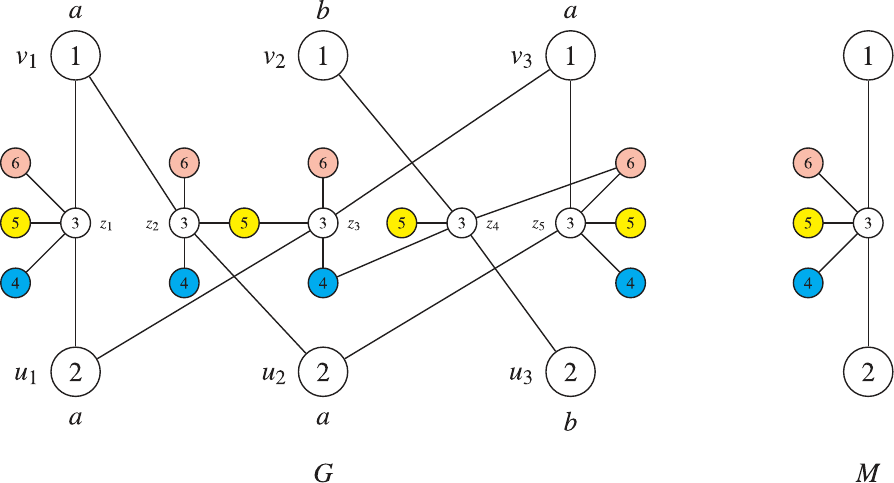}
      \caption{\label{fig:disjointmotifs} Graph and motif from
        sequences $R = \{\textit{aba}, \textit{aab}\}$. Vertices with
        colors 1, 2, or 3 are transversal vertices, those with colors
        4 and 5 are sequence vertices, and those with color 6 are core
        vertices.}
    \end{center} 
  \end{figure}
  
  Graphs $G$ and $M$ can be obtained in polynomial time with respect
  to the size of $R$. Then we must show that there exist at least $k$
  disjoint occurrences of $M$ in $G$ if and only if there exists a
  common subsequence in $R$ of length at least~$k$.

  Consider a vertex set $K$ such that $K \subseteq V_G$. Denote by
  $K^Z$ the set $K \cap V^Z$, for $Z \in \{S, C, X\}$.
  
  Notice that there exists an obvious correspondence between columns
  in $R$ and occurrences of $M$ in $G$, i.e., for each column
  $\bfit{j} = (j_1, \ldots, j_p)$ in $R$, there exists an occurrence
  $J$ in $G$ such that $J^S = \{v_1^{j_1}, \ldots,
  v_p^{j_p}\}$. Therefore, to show that there exists a common
  subsequence in $R$ of length at least $k$ if and only if there exist
  at least $k$ disjoint occurrences of $M$ in $G$, it is enough to
  prove that two columns $\bfit{j} = (j_1, \ldots, j_p)$ and $\bfit{k}
  = (k_1, \ldots, k_p)$ in $R$ are non-crossing if and only if
  subgraphs $J$ and $K$ are disjoint in $G$.

  Suppose that $\bfit{j}$ and $\bfit{k}$ are non-crossing. Then, $j_i
  \not= k_i$ for each $i$ which implies that $v_i^{j_i} \not=
  v_i^{k_i}$ for each $i$ and therefore, $J^S \cap K^S =
  \emptyset$. By construction, all core vertices in $J$ are adjacent
  to two vertices in $J^S$ and all core vertices in $K$ are adjacent
  to two vertices in $K^S$, which implies, since $J^S \cap K^S =
  \emptyset$, that $J^C \cap K^C = \emptyset$. Below, until the end
  of this paragraph, we are going to show that an arbitrary vertex $u
  \in J^X$ does not belong to $K^X$, which implies that $J^X \cap K^X
  = \emptyset$. By construction, since $u \in J^X$, it follows that
  $u$ must be connected to vertex $\mathbf{v}_i^y \in J^C$, where $y =
  (v_i^{j_i}, v_{i+1}^{j_{i+1}})$ for some $i$. Since $J^C \cap K^C =
  \emptyset$ and $u$ is connected to $\mathbf{v}_i^y \in J^C$, if $u$
  is a leaf, then $u$ is not connected to any vertex in $K^C$ which
  implies that $u \not\in K^X$. Then, assume that $u$ is not a
  leaf. It follows by construction that $u$ is connected to another
  vertex $\mathbf{v}_i^z \in V^C$, where $z = (v_i^m, v_{i+1}^n) \not=
  y$, for some pair $m, n$. This implies that $\mathbf{v}_i^y$ and
  $\mathbf{v}_i^z$ are opposing in $i$, and thus $j_i < m$ and
  $j_{i+1} > n$, or $j_i > m$ and $j_{i+1} < n$. Therefore, since
  $\bfit{j}$ and $\bfit{k}$ are non-crossing columns in $R$, we have
  that $s_{i, m} \not\in \bfit{k}$ or $s_{i+1, n} \not\in \bfit{k}$
  and consequently, $v_i^m \not\in K$ or $v_{i+1}^n \not\in
  K$. However, by construction, if an arbitrary core vertex
  $\mathbf{v}$ belongs to an occurrence, then each vertex $u$ adjacent
  to $\mathbf{v}$ also belongs to this occurrence, which implies that
  $u \not\in K$. Since $J^S \cap K^S = J^C \cap K^C = J^X \cap K^X =
  \emptyset$, we conclude that $J$ and $K$ are disjoint in $G$.
 
  Conversely, suppose that $\bfit{j}$ and $\bfit{k}$ are crossing in
  $R$. Then, for some integer $i$, we have that $j_i \leq k_i$ and
  $j_{i+1} \geq k_{i+1}$ or $j_i \geq k_i$ and $j_{i+1} \leq
  k_{i+1}$. W.l.o.g., suppose that $j_i \leq k_i$ and $j_{i+1} \geq
  k_{i+1}$. If $j_i = k_i$ or $j_{i+1} = k_{i+1}$, we have that
  $v_i^{j_i}$ or $v_{i+1}^{j_{i+1}}$ belongs to $J^S \cap K^S$, which
  implies that $J$ and $K$ are not disjoint. Then, we assume that $j_i
  < k_i$ and $j_{i+1} > k_{i+1}$. By construction, we have
  $\mathbf{v}_i^y = (v_i^{j_i}, v_{i+1}^{j_{i+1}})$ in $J^C$ and
  $\mathbf{v}_i^z = (v_i^{k_i}, v_{i+1}^{k_{i+1}})$ in $K^C$ opposing
  in $i$ and there exists a vertex $u$ in $V^X$ adjacent to both
  $\mathbf{v}_i^y$ and $\mathbf{v}_i^z$. Since all vertices adjacent
  to $\mathbf{v}_i^y$ belong to $J$ and all vertices adjacent to
  $\mathbf{v}_i^z$ belong to $K$, it follows that $u \in J \cap K$,
  which implies that $J$ and $K$ are not disjoint.
\end{proof}

Finally, we show NP-hardness of \textsc{Common-$k$-Tree}.

\begin{theorem} \label{theo:common-k-tree}
  Problem \textsc{Common-$k$-Tree} is NP-complete even if $G$ is a tree with
  maximum degree three.
\end{theorem}

\begin{proof}
  Fellows \textit{et al.}~\cite{FFHV2007} showed that the following
  problem is NP-complete: Given a vertex-colored tree $G$ with maximum
  degree three and a set of colors $\colorset$, decide if there exists
  a colorful tree $T \prec G$ such that $|V_T| = |\colorset|$. This
  problem can be formulated as a particular instance of
  \textsc{Common-$k$-Tree}($G, G, |\colorset|$). Thus,
  \textsc{Common-$k$-Tree} is NP-complete even if $G$ is a tree with
  maximum degree three.
\end{proof}

\section{\label{sec:sear} Search algorithms}

In this section we present algorithms for searching various types of
motifs in vertex-colored graphs. The following section presents an
efficient algorithm such that we are given a vertex-colored graph $G$
and a colorful tree $T$ and we want to find a subgraph of $G$
isomorphic to $T$. We also want to find all occurrences of a colorful
motif $T$ in a vertex-colored graph $G$. Next, we first introduce a
data structure simpler than $G$, called maximum clean graph, from
which it is possible to obtain all subgraphs. We can then obtain more
smoothly all colorful motifs in a clean subgraph later. Finally, we
describe a linear time algorithm for computing the number of
occurrences of colorful motifs in a vertex-colored graph without
enumerating all occurrences.

\subsection{\label{sec:subg} Finding subgraphs isomorphic to a colorful tree}

If the given colorful motif $M$ for the \textsc{Subgraph-Motif}
problem is a tree, we can solve the problem efficiently. In the
following we present a linear time algorithm for
\textsc{Subgraph-Motif} for this particular type of instances.

Let $G$ be a graph and $T$ be a tree which is a colorful motif.
Algorithm~\textsc{tcg} starts by identifying an arbitrary leaf $x$ of
$T$ in Step 1. Step 2 collects in a set $A$ all vertices of $G$ with
the color of $x$. Then, if $T$ is a graph with only one vertex,
Algorithm~\textsc{tcg} verifies whether the set $A$ is empty,
returning yes or no in Step 3. Otherwise, Step 4 takes $y$ as the only
vertex adjacent to $x$ in $T$. Then, Step 5 builds a set $B$ with all
vertices in $G$ with the color of $y$ and with no neighbors with the
color of $x$. Finally, Step 6 is a recursive call to \textsc{tcg}
removing vertices in $A$ and $B$ from $G$ and the leaf $x$ from
$T$. We suppose that graphs $G$ and $T$ are given to
Algorithm~\textsc{tcg} as adjacency lists.

\begin{algorithm}
\caption{~~\textsc{tcg}$(G,T)$} \label{alg:TCG}
\begin{algorithmic}[1]

  \footnotesize
  
  \REQUIRE{graph $G$, colorful tree $T$}
  \ENSURE{\textsc{yes/no}, whether $T \prec G$ or not}

  \smallskip
  
  \STATE let $x$ be some leaf of $T$ 
  
  \STATE let $A := \{a \in V_G : \colorf(a) = \colorf(x)\}$ 
  
  \IF{$T$ is trivial, i.e., it has one single vertex}

    \RETURN $A \neq \emptyset$

  \ENDIF
    
  \STATE let $y$ be the vertex such that $xy \in E_T$
  
  \STATE let $B := \{ b \in V_G : \colorf(b) = \colorf(y) \text{ and
    for all } b' \in V_G \text{ such } \text{ that } \newline bb'\in
  E_G, \text{ then } \colorf(b') \neq \colorf(x)\}$
  
  \RETURN \textsc{tcg}($G - (A \cup B), T - \{x\}$)
  
\end{algorithmic}
\end{algorithm}

The following shows the correctness of Algorithm~\textsc{tcg}.

\begin{theorem} \label{theo:TCG-correct}
Given a vertex-colored graph $G$ and a colorful tree $T$,
Algorithm~\textsc{tcg} returns correctly whether $G$ has a subgraph
isomorphic to $T$ or not.
\end{theorem}

\begin{proof}
  Denote by $G_i$ and $T_i$ the input graphs in the $i$th recursive
  call to Algorithm~\textsc{tcg}, for $i \geq 1$, where $G_1 = G$ and
  $T_1 = T$. (Since there is no change of colors in the graphs $G$ and
  $T$ for any call of the algorithm, we drop those indices for the
  color function.) To establish the correctness, consider a sequence
  of calls $1, 2, \ldots, |V_T|$ of Algorithm~\textsc{tcg}.  Then, it
  is enough to show that
  \[
  T_{i-1} \prec G_{i-1} \text{ if and only if } T_i \prec G_i\:,
  \]
  for any recursive call $i > 1$. If $|V_T| = 1$, for some recursive
  call $i$, Algorithm~\textsc{tcg} runs the three first steps and
  returns correctly whether $G$ has a subgraph isomorphic to~$T$.
  
  Suppose first that $T_{i-1} \prec G_{i-1}$ at the $(i-1)$st call of
  Algorithm~\textsc{tcg}. Since we have a subsequent call $i$, there
  exists a subgraph $G'$ of $G_{i-1}$ such that $T_{i-1} \cong
  G'$. Still, there exists a vertex $a$ in $G_{i-1}$ such that
  $\colorf(a) = \colorf(x)$ for a chosen leaf $x$ in $T_{i-1}$, $a$ is
  adjacent to a vertex $b$ in $G'$ with color $\colorf(b) =
  \colorf(y)$, and $xy$ is an edge in $T_{i-1}$. Then, $G' - \{a\}
  \subseteq G_i$ and $T_i = T_{i-1} - \{x\} \cong G' - \{a\}$. Thus,
  $T_i \prec G_i$.

  Now, suppose that $T_i \prec G_i$ at the $i$th call of
  Algorithm~\textsc{tcg}.  Then, there exists a subgraph $G'$ of $G_i$
  such that $T_i \cong G'$. Let $b$ be a vertex in $G'$ and $y$ be a
  vertex in $T_i$ such that $\colorf(b) = \colorf(y)$ and $y$ is
  adjacent to the removed vertex $x$ in $T_{i-1}$. Since $b$ was not
  removed in $G_i$, there should exist at least one vertex $a$ in
  $G_{i-1}$ such that $\colorf(a) = \colorf(x)$ and $ab \in
  E_{G_{i-1}}$.  Thus, the subgraph $G''= G' + ab$ is a subgraph of
  $G_{i-1}$ and $G'' \cong T_{i-1}$.  Hence, $T_{i-1} \prec G_{i-1}$.
\end{proof}

A quick inspection of the pseudocode of Algorithm~\textsc{tcg} gives
us a quadratic running time in the size of $G$. Nevertheless, some
preprocessing and a more careful analysis lead to a linear running
time as the following result states.

\begin{theorem} \label{theo:TCG-time}
Given a vertex-colored graph $G$ and a colorful tree $T$,
Algorithm~\textsc{tcg} returns correctly whether $G$ has a subgraph isomorphic
to $T$ in time linear in the size of $G$.
\end{theorem}

\begin{proof}
  We start by preprocessing graphs $G$ and $T$ in linear
  time. Visiting all vertices in $V_T$ we have a list of leaves
  $\mathbf{L}_T$ of $T$. Moreover, for each color $\gamma$ in
  $\colorset$ we create a list of vertices $\mathbf{C}_\gamma$ in
  $V_G$, such that each vertex $v$ in $\mathbf{C}_\gamma$ satisfies
  $\colorf(v) = \gamma$.

  Steps~1--4 can then be performed in $O(1)$ time: Step~1 can be done by
  obtaining a vertex from $\mathbf{L}_T$ and Step~2 by obtaining
  vertices from $\mathbf{C}_\gamma$ not removed from $V_G$ such that
  $\gamma = \colorf(x)$. Steps~3 and~4 are trivial.

  The overall time spent in Step~5 across all recursive calls is
  linear in the size of $G$. Notice that whenever we remove $x$ from
  $T$, we have to add $y$ to $\mathbf{L}_T$ if $y$ becomes a leaf.

  Step~6 is the most tricky one. First of all, we visit vertices in
  $A$ and mark their neighbors. Then we add to set $B$ each unmarked
  vertex in $\mathbf{C}_\gamma$ such that $\gamma = \colorf(y)$
  (according to Step~5). Finally, neighbors of vertices in $A$ are
  unmarked. Notice that in Step~6 a vertex can be analyzed many times
  across the recursive calls. However, if a vertex is removed at some
  point, it happens at most once. On the other hand, a vertex is
  analyzed and not removed at most a number of times equal to the
  number of adjacent edges. This means that the vertex is adjacent to
  at least one vertex in $A$. Hence, the overall time spent in Step~6,
  considering all the recursive calls, is linear in the size of $G$.

  Therefore, the overall running time of Algorithm~\textsc{tcg} is
  linear in the size of the instance $G$.
\end{proof}

In general, we are not only interested in whether there exists a
subgraph of $G$ isomorphic to a colorful motif $T$, but also in
finding one or all different such subgraphs.  Actually, finding one
subgraph of a graph $G$ isomorphic to a colorful motif $T$ can be
achieved by a minor change of Algorithm~\textsc{tcg}. Namely, we
maintain a copy $G'$ of $G$ and a pointer array $p$ representing the
predecessors of each vertex $v$ in $V_{G'}$ responsible for $v$ to
become a candidate belonging to some subgraph of $G$ isomorphic to
$T$. Initially $p[v]$ is \textsc{nil} for all $v \in V_{G'}$. During
Step 5, for each non-removed vertex $v$ in $V_{G}$ with $\colorf(v) =
\colorf(y)$ the list $p[v]$ is updated, containing each vertex $w$ in
$V_{G'}$ adjacent to $v$ such that $\colorf(w) = \colorf(x)$. Hence, a
motif can easily be found after the execution of
Algorithm~\textsc{tcg} using the array $p$. The second task, i.e.,
finding all occurrences of a colorful motif $T$ in $G$, is realized by
algorithms described in the following.

\subsection{\label{sec:maxi} Maximum clean subgraphs}

Finding all occurrences of a colorful motif $T$ in a vertex-colored
graph $G$ can be a highly time consuming task, since the number of
occurrences of some $T$ in $G$ may in fact be exponential in the worst
case. See an example in Fig~\ref{fig-exponentialexample}. Hence,
finding all isomorphic subgraphs of $T$ in $G$ can be prohibitive in
the general case. Nevertheless, it is useful to have a structure
simpler than $G$ from which it is possible to obtain all subgraphs
of $G$ isomorphic to $T$. This structure is described below.

\medskip

\begin{figure}[h]%
  \centering%
  \includegraphics{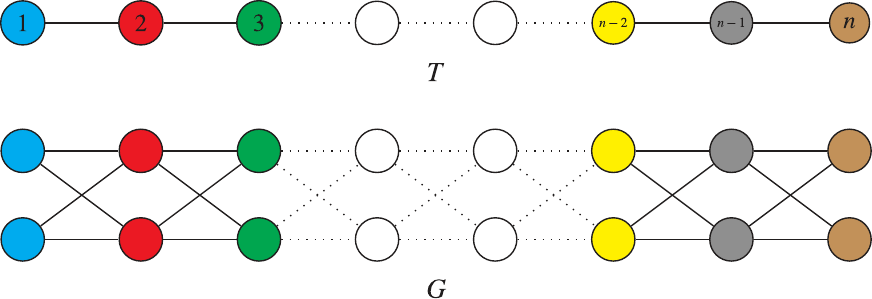}
  \caption{\label{fig-exponentialexample} A simple path $T$ with $n$
    vertices, where the first vertex has color 1 (cyan), the second
    has color 2 (red), and so on. Graph $G$ has two vertices for each
    color in $\colorset$ and each vertex with color $i$, $1 \leq i <
    n$, is adjacent to both vertices of color $i + 1$ in $G$. Thus, we
    have $2^n$ occurrences of $T$ in $G$ and therefore there exist
    $\Omega(2^n)$ occurrences of $T$ in $G$ in the worst case.}
\end{figure}

We say that a graph $H$ is a \emph{clean subgraph regarding $T$} if
each vertex and each edge in $H$ is a vertex, respectively an edge, of
a subgraph of $H$ isomorphic to $T$. For a vertex-colored graph $G$,
we say that $H$ is the \emph{maximum clean subgraph of $G$ regarding
  $T$} if $H$ is a clean subgraph regarding $T$ and any subgraph of
$G$ isomorphic to $T$ is also a subgraph of $H$.
Algorithm~\textsc{mcg} finds the maximum clean subgraph of $G$
regarding $T$.

\begin{algorithm}
\caption{~~\textsc{mcg}$(G, T)$} \label{alg:MCG}  
\begin{algorithmic}[1]

  \footnotesize
  
  \REQUIRE{graph $G$, colorful tree $T$}
  \ENSURE{maximum clean subgraph $H$ of $G$ regarding $T$}

  \smallskip
  
  \STATE $H := G$
  \FOR{each $v \in V_H$} \label{mcg:line:for1s}
    \IF{$\colorf(v)$ is not a color in $T$}
       \STATE $H := H - \{ v \}$ \label{mcg:line:for1f}
    \ENDIF
  \ENDFOR
  \FOR{each $vv' \in E_H$}
    \IF{there exists no edge $ww'$ in $T$ such that $\colorf(v) =
      \colorf(w)$ and $\colorf(v') = \colorf(w')$} 
      \STATE $H := H - \{ vv' \}$ \label{mcg:line:for2f}
    \ENDIF
  \ENDFOR     
  \WHILE{there exist $v \in V_H$ and $w, w' \in V_T$ such that
    $\colorf(v) = \colorf(w)$ and $ww' \in E_T$, but there does not
    exist $v' \in V_H$ such that $vv' \in E_H$ and  $\colorf(v') =
    \colorf(w')$} \label{mcg:line:while1s}  
    \STATE $H := H - \{ v \}$ \label{mcg:line:while1f}
  \ENDWHILE
  \RETURN $H$
\end{algorithmic}
\end{algorithm}

Since vertices and edges removed from $H$ by Algorithm~\textsc{mcg}
cannot be vertices and edges of any subgraph of $G$ isomorphic to $T$,
we can prove that $H$ is the maximum clean subgraph of $G$ regarding
$T$. We do this through Lemmas~\ref{lem:x-subgraph}--\ref{lem:z-subgraph}.
First we can verify the following auxiliary result.

\begin{lemma} \label{lem:H-V}
Given a graph $G$ and a colorful tree $T$, let $H$ be a
graph obtained by Algorithm~\textsc{mcg}. Let $v \in V_T$ be a
leaf and let $\mathcal{V} = \{w \in V_H : \colorf(w) =
\colorf(v)\}$ be the set of vertices in $H$ with color
$\colorf(v)$. Then, \textsc{mcg}$(H, T - \{v\})$ returns $H -
\mathcal{V}$.
\end{lemma}

\begin{proof}
Since $T - v$ is a colorful motif, lines 2--4 of Algorithm~\textsc{mcg}
remove all vertices in $H$ with color $\colorf(v)$, i.e., the vertices
in $\mathcal{V}$.
\end{proof}

\begin{lemma} \label{lem:x-subgraph}
Given a graph $G$ and a colorful tree $T$, let $H$ be a graph obtained
by Algorithm~\textsc{mcg}. Then, any vertex $x$ in $H$ is a vertex of
a subgraph $\widetilde{G}$ in $H$ isomorphic to $T$.
\end{lemma}

\begin{proof}
We prove the lemma by induction on $|V_T|$.

Suppose that $|V_T| \leq 1$. If $|V_T| = 0$, after lines 1--4, $H$ has
no vertex and the proof is completed. If $|V_T| = 1$, after lines
1--4, the proof is also completed since each vertex in $H$ is
isomorphic to $T$.

If $|V_T| > 1$ then $T$ has at least two different leaves. Let $x$ be
an arbitrary vertex in $H$. Since $T$ is a colorful motif, it follows
that, for one of these leaves, say $v$, we have $\colorf(x) \not=
\colorf(v)$. Let $w$ be the neighbor of $v$ in $T$. From
Lemma~\ref{lem:H-V}, we have that \textsc{mcg}$(H, T - \{v\}) = H -
\mathcal{V}$, where $\mathcal{V} = \{u \in V_H : \colorf(u) =
\colorf(v)\}$. Since $\colorf(x) \not= \colorf(v)$, then $x \in V_{H -
  \mathcal{V}}$. By the induction hypothesis, $x$ is a vertex of a
subgraph $\widetilde{G}' \cong T - \{v\}$ in $H - \mathcal{V}$. Let
$w'$ be a vertex in $\widetilde{G}'$ such that $\colorf(w') =
\colorf(w)$. Since $w' \in V_H$, $w'$ must have a neighbor $v'\in V_H$
with color $\colorf(v)$, otherwise $w'$ would have been removed in
line~\ref{mcg:line:while1f} of Algorithm~\textsc{mcg}($G,T$). It
follows that $\widetilde{G} = \widetilde{G}' + \{ v'w' \} \cong T$,
$\widetilde{G}$ is a subgraph of $H$, and $x \in V_{\widetilde{G}}$.
\end{proof}

\begin{lemma} \label{lem:xy-subgraph}
Given a vertex-colored graph $G$ and a colorful tree $T$, let $H$ be a
vertex-colored graph obtained by Algorithm~\textsc{mcg}. Then, any
edge $xy$ in $H$ is an edge of a subgraph $\widetilde{G}$ of $H$
isomorphic to $T$.
\end{lemma}

\begin{proof}
From Lemma~\ref{lem:x-subgraph}, there exist subgraphs
$\widetilde{G}_x \cong T$ and $\widetilde{G}_y \cong T$ in $H$
containing vertices $x$ and $y$, respectively. Then, there exist
vertices $x'$ and $y'$ with colors $\colorf(x)$ and $\colorf(y)$,
respectively, and $xy'$ is an edge in $\widetilde{G}_x$ and $x'y$ is
an edge in $\widetilde{G}_y$. Let $G'_x$ be the component of
$\widetilde{G}_x - \{xy'\}$ containing vertex $x$, and $G'_y$ be the
component of $\widetilde{G}_y - \{x'y\}$ containing vertex $y$. The
graph $\widetilde{G}$ which is the union of $G'_x$ and $G'_y$ plus the
edge $xy$ is a subgraph of $H$ isomorphic to $T$.
\end{proof}

\begin{lemma}\label{lem:z-subgraph}
  At the start of each iteration of the loop of
  lines~\ref{mcg:line:while1s}--\ref{mcg:line:while1f} in
  Algorithm~\textsc{mcg}, all vertices in~$V_G \setminus V_H$ and all
  edges in~$E_G \setminus E_H$ do not belong to any subgraph of $G$
  isomorphic to $T$.
\end{lemma} 

\begin{proof}
  Since, in any subgraph of $G$ isomorphic to $T$, the color of a
  vertex must be a color of a vertex in $T$ and the colors of
  endpoints of an edge must be the colors of endpoints of an edge in 
  $T$, we have that steps~\ref{mcg:line:for1f}
  and~\ref{mcg:line:for2f} remove only vertices and edges that do not
  belong to any subgraph of $G$ isomorphic to $T$. Thus, the statement
  is true prior to the first iteration of the loop of
  lines~\ref{mcg:line:while1s}--\ref{mcg:line:while1f}.

  Suppose that the statement is true before an iteration of the loop
  of lines~\ref{mcg:line:while1s}--\ref{mcg:line:while1f}.  Let $v$ in
  $V_H$ and $w, w'$ in $V_T$ be vertices such that $\colorf(v) =
  \colorf(w)$, $ww' \in E_T$, and there does not exist $v'$ in $V_H$
  such that $vv' \in E_H$ and $\colorf(v') = \colorf(w')$. Since, by
  hypothesis, all vertices in~$V_G \setminus V_H$ and all edges
  in~$E_G \setminus E_H$ do not belong to any subgraph of $G$
  isomorphic to $T$, we have that vertex $v$ and edges incident to it
  do not belong to any subgraph of $G$ isomorphic to $T$ either. Since
  they are precisely the vertices and edges removed in $H$, the
  statement is still true before the next iteration.
\end{proof}

\begin{theorem}
  Given a vertex-colored graph $G$ and a colorful tree $T$, let $H$ be
  a graph given by Algorithm~\textsc{mcg}. Then, $H$ is the maximum
  clean subgraph of $G$ regarding~$T$.
\end{theorem}

\begin{proof}
  From Lemmas~\ref{lem:x-subgraph} and~\ref{lem:xy-subgraph}, since
  all vertices and edges belong to a subgraph of $G$ isomorphic to
  $T$, we have that $H$ is a clean subgraph. On the other hand, as a
  consequence of Lemma~\ref{lem:z-subgraph}, all vertices and edges
  removed by Algorithm~\textsc{mcg} are not in any subgraph of $T$ in
  $G$, and thus $H$ is the maximum clean subgraph of $G$ regarding
  $T$.
\end{proof}

We show now the running time of Algorithm~\textsc{mcg}.

\begin{theorem}
  Given a vertex-colored graph $G$ and a colorful tree $T$,
  Algorithm~\textsc{mcg} runs in time cubic in the number of vertices
  of $G$.
\end{theorem}

\begin{proof}
  Consider that $G$ and $T$ are given by their adjacency matrices and
  the set of colors is given by an ordered array.

  Line~1 of Algorithm~\textsc{mcg} can be implemented in $O({V_G}^2)$
  time since $H$ is a copy of $G$. Notice that we adopt the notation
  of~\cite{CLRS09} where inside asymptotic notation, the symbol $V$
  denotes $|V|$ and the symbol $E$ denotes $|E|$. The color of each
  vertex in $H$ can be checked in $O(1)$ time, and each vertex can be
  removed in $O(V_G)$ time, which implies that lines~2--4 can be
  performed in $|V_G| \cdot (O(1) + O(V_G)) = O({V_G}^2)$ time. We can
  verify whether colors of the extremities of an edge in $H$ are also
  the colors of the extremities of an edge in $T$ in $O(1)$ time, and,
  if necessary, its deletion can also be done in $O(1)$ time which
  implies, since we have $O({V_G}^2)$ edges, that Step~3 can be
  performed in $O({V_G}^2) \cdot (O(1) + O(1)) = O({V_G}^2)$ time. In
  Step~4, deciding whether a vertex should be removed or not, can be
  done spending $O({V_G}^2)$ time, and the removal can be done in
  $O(V_G)$ time, when necessary, implying that the total time spent on
  Step~4 is $|V_G| \cdot (O({V_G}^2) + O(V_G)) =
  O({V_G}^3)$. Therefore, the running time of Algorithm~\textsc{mcg}
  is $O({V_G}^3)$.
\end{proof}

\subsection{\label{sec:allo} Finding all occurrences of a colorful motif in a clean subgraph}

We present now Algorithm~\textsc{All-Colorful} for finding all the
colorful motifs in a given maximum clean subgraph. Despite this task
being superpolynomial, sometimes it is useful to find all motifs in
order to choose those we are interested in. The input of the algorithm
is the maximum clean subgraph $H$, as defined in the previous
section. The colorful tree $T$ is implicit.  If $V_H = \emptyset$,
then there does not exist any subgraph of $H$ isomorphic to $T$
(except if $V_T = \emptyset$). Thus, we assume that $H$ has at least
one vertex.

\begin{algorithm}
\caption{\label{alg:all-colorful} \textsc{All-Colorful}$(H)$}
\begin{algorithmic}[1] 

  \footnotesize
  
  \REQUIRE{maximum clean subgraph $H$}
  \ENSURE{set $\mathcal{G}$ of all subgraphs of $H$ isomorphic to a
    colorful motif} 

  \smallskip

  \IF{$|\{ \colorf(a) : a \in V_H \}| = 1$}
  
    \STATE $\mathcal{G} := \{(a,\emptyset) : a \in V_H\}$
  
  \ELSE
    
    \STATE let $\alpha$ be a color of a vertex of $H$ whose neighbors
    have all the same color $\beta$
    \STATE $\mathcal{A} := \{a \in V_H : \colorf(a) =
    \alpha\}$ 
      
    \STATE $\mathcal{G} := \emptyset$
      
    \STATE $\mathcal{H} := \textsc{All-Colorful}(H - \mathcal{A})$
      
    \FOR{each vertex $a$ in $\mathcal{A}$}
      
      \STATE $\mathcal{B} := \{b \in V_H : \colorf(b) = \beta$
      and $ab \in E_H\}$
        
      \FOR{each $\widetilde{G}$ in $\mathcal{H}$ and $b' \in
        \mathcal{B} \cap V_{\widetilde{G}}$}
        
        \STATE $\mathcal{G} := \mathcal{G} \cup \{\widetilde{G} +
        \{ab'\}\}$ 
          
      \ENDFOR
          
    \ENDFOR
        
  \ENDIF
    
  \RETURN $\mathcal{G}$
  
\end{algorithmic}
\end{algorithm}

Initially, we can show the following easy result.

\begin{lemma}\label{lema:aux}
Let $T$ be a colorful tree such that $|V_T| = 1$. Then, for each $v
\in V_H$, we have $\{v\} \in \mathcal{G}$, $\{v\}$ is a subgraph of
$H$ and $\{v\} \cong T$.
\end{lemma}
\begin{proof}
  Since $H$ is a clean subgraph regarding $T$, the lemma holds
  immediately.
\end{proof}

The following result shows that \textsc{All-Colorful}$(H)$ returns the
set of all subgraphs of $H$ isomorphic to $T$.

\begin{theorem}
Let $G$ be a vertex-colored graph and $T$ be a colorful tree. Suppose
that $H$ is the maximum clean subgraph returned by \textsc{mcg}$(G,
T)$, $\mathcal{G}$ is the set of graphs returned by
\textsc{All-Colorful}$(H)$, and $\widetilde{G}'$ is a subgraph of
$H$. Then, $\widetilde{G}' \in \mathcal{G}$ if and only if
$\widetilde{G}' \cong T$.
\end{theorem}

\begin{proof}
  We prove the theorem by
  induction on the number of vertices in $T$. If $|V_T| = 1$, then,
  since \textsc{mcg}($G, T$) returns $H$, all vertices in $H$ have the
  same color. Thus, from lines~1--2 of the Algorithm~\textsc{All-Colorful},
  we have that $a \in \mathcal{G}$ for each $a
  \in V_H$. On the other hand, from Lemma~\ref{lema:aux}, $\{a\} \cong
  T$ for each $a \in V_H$. Therefore the theorem holds if $|V_T| =
  1$.

  Suppose that $|V_T| > 1$ and $\widetilde{G}' \in \mathcal{G}$. Then,
  $\mathcal{A} \neq \emptyset$ and $\mathcal{B} \neq \emptyset$. Since
  $T$ is a tree and $|V_T| > 1$, $T$ has at least a vertex $v$ of
  degree one and since $H$ is a clean subgraph there exists a vertex
  $v'$ in $V_H$ such that $\colorf(v')=\colorf(v)$, i.e., $\mathcal{A}
  \neq \emptyset$. Likewise, each vertex $w$ in $V_T$ adjacent to $v$
  implies that $c(w)=\beta$, for some color $\beta$. Hence, as $H$ is
  a clean subgraph, then $\mathcal{B} \neq \emptyset$. It follows
  that, by line~11, $\widetilde{G}' = \widetilde{G} + \{ab'\}$, where
  $b' \in \mathcal{B}$ is a vertex in $\widetilde{G}$ in $\mathcal{H}$
  and $a \in \mathcal{A}$. Let $u, u'$ be vertices in $V_T$ such that
  $\colorf(u) = \colorf(a)$ and $\colorf(u')= \colorf(b')$. Since $H$
  is a clean subgraph regarding $T$, it follows that $u$ and $u'$ are
  neighbors in $T$. Further, since $H$ is the maximum clean subgraph
  regarding $T$ and every neighbor of a vertex with color $\colorf(u)$
  has color $\colorf(u')$, then $u$ is a leaf in $T$. From
  Lemma~\ref{lem:H-V}, $H - \mathcal{A}$ is the maximum clean subgraph
  regarding $T - \{u\}$. By the induction step, $\widetilde{G}$ is a
  subgraph of $H - \mathcal{A}$ isomorphic to $T - \{u\}$. Thus,
  $\widetilde{G}' = \widetilde{G} + \{ab'\}$ is isomorphic to $T = T -
  \{u\} + \{uu'\}$. Consequently, if $\widetilde{G}' \in \mathcal{G}$,
  then $\widetilde{G}'$ is isomorphic to $T$, for $|V_T| > 1$.

  Conversely, suppose that $\widetilde{G}'$ is a subgraph of $H$
  isomorphic to $T$. Let $a$ and $b'$ be vertices in
  $\widetilde{G}'$. By the inductive hypothesis,
  $\textsc{All-Colorful}(H - \mathcal{A})$ returns $\widetilde{G} =
  \widetilde{G}' - \{a\}$. Then, by line~11 we have $\widetilde{G}' =
  \widetilde{G}' + \{ab'\} \in \mathcal{G}$.
\end{proof}

For a vertex-colored graph $G$ and a colorful motif $T$, we denote the
\emph{number of occurrences of~$T$ in $G$} by $\eta(G,T)$. Clearly, an
algorithm for enumerating motifs has running time $\Omega(\eta(G, T))$
and, since $\eta(G, T)$ can be superpolynomial, an algorithm for
enumerating motifs has exponential running time in the worst case. On
the other hand, considering all recursive calls, the time consumed
by~\textsc{All-Colorful} is dominated by Step~11 and, since each edge
can be added to at most $\eta(G, T)$ subgraphs, the total running time
of the algorithm is $O(E_G \cdot \eta(G, T))$, which implies that the
exponential part lies only in the output size. With this observation,
we have the following immediate result.

\begin{corollary}
  If $\eta(G, T)$ is polynomial in the size of $G$, then
  Algorithm~\textsc{All-Colorful} also runs in polynomial time in the
  size of $G$.
\end{corollary}

\subsection{\label{sec:numb} Finding the number of occurrences of a colorful motif in a graph}

Let $G$ be a graph and $T$ be a colorful tree. Notice that $\eta(G,T)$
can be easily obtained by finding all occurrences of $T$ in $G$. As
seen in the previous section, this procedure runs in exponential time
in the worst case, since the number of occurrences of $T$ in $G$ may
be exponential. Thus, we describe now a linear time algorithm for
computing $\eta(G, T)$ without enumerating all occurrences.

Choose an arbitrary vertex $r$ in $T$ and, from now on, deal with $T$
as a rooted tree with root $r$. Let $u$ and $v$ be vertices in $T$. If
$u$ belongs to a path from $r$ to $v$, we say that $v$ is a
\emph{descendant} of $u$ in $T$.  We define $T_u$ as the rooted
subtree of $T$ with root $u$ plus all descendants of $u$.  We also
define the \emph{direct descendants of $u$} as the set of vertices
$A_u := \{ v \in T_u: uv \in E_T \}$. Moreover, $L(i)$ denotes a list
of vertices of $G$ with color $i$.  Now, for the next definitions, we
consider a vertex $u$ in $T$ and a vertex $z$ in $G$ such that
$\colorf(u) = \colorf(z)$. We then define a set of vertices
$\text{adj}_z := \{y \in V_G : yz \in E_G \text{ and } \colorf(y) =
\colorf(w) \text{ for some vertex $w \in A_u$} \}$, and an integer
$\text{N}_z := |\{ S : S \text{ is an occurrence of $T_u$ in $G$ and
  $z \in S$} \}|$.

Thus, considering $u \in V_T$ and $z \in V_G$ such that $\colorf(u) =
\colorf(z)$, we have

\begin{align}
  \text{N}_z &= \left\{
  \begin{array}{ll}
    1\:, & \text{if $u$ is a leaf in $T$}\:,\\
    \displaystyle \prod_{v \in A_u} 
    \bigg( 
    \sum_{y \in \text{adj}_z : \colorf(y) = \colorf(v)} \!\!\!\! \text{N}_y 
    \bigg)
    \:, & \text{otherwise}\:.
  \end{array} 
  \right. \label{exp:numberOfOcurrences}
\end{align}

Let $v_1 (= r), v_2, \ldots, v_{|V_T|}$ be a list of visited vertices
of $T$ in a breadth-first search with an arbitrary source $r$. Assume
that $\colorf(v_i) = i$. We use the reverse order of this list to
compute $\text{N}_z$, according to~(\ref{exp:numberOfOcurrences}),
for each vertex $z$ in $L(i)$ and color $i$. After $\text{N}_z$ has
been computed for each $z$ in $V_G$, we return $\sum_{z \in L(1)}
\text{N}_z$. See Fig~\ref{fig-numberOfMotifs}.

\begin{figure}[h!]%
  \centering%
  \includegraphics{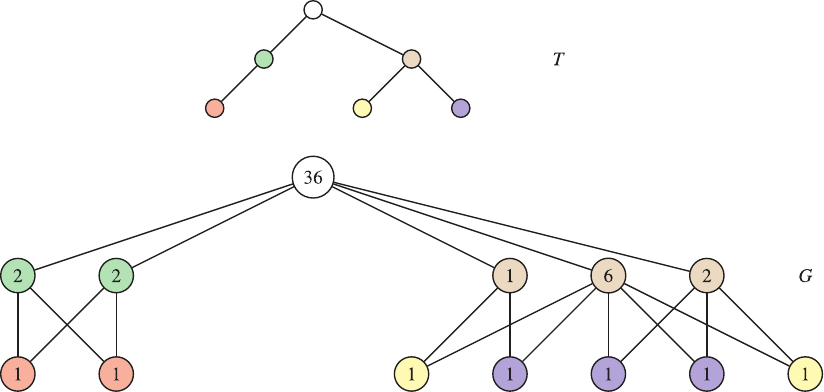}
  \caption{\label{fig-numberOfMotifs} A colorful tree $T$ rooted at
    the white vertex. In $G$, each vertex $z$ is labeled with
    $\textrm{N}_z$. Notice that $\textrm{N}_z$ represents the number
    of occurrences of $T_{u}$ in $G$, where $\colorf(u) =
    \colorf(z)$.}
\end{figure}
 
The procedure to compute $\eta(G, T)$ is summarized in
Algorithm~$\eta$.

\begin{algorithm}[!h]
\caption{\label{alg:numocc} $\eta(G, T)$}
\begin{algorithmic}[1] 

  \footnotesize
  
  \REQUIRE{graph $G$, colorful tree $T$}

  \ENSURE{number of occurrences of $T$ in $G$}

  \smallskip

  \STATE let $r$ be an arbitrary vertex in $T$

  \STATE let $v_1, v_2, \ldots, v_{|V_T|}$ be a list of visited
  vertices of $T$ in a breadth-first traversal with source
  $r$ \label{alg-nOfOccur:st1}

  \STATE let $\colorf(v_i) := i$ 

  \FOR{each color $i = |V_T|$ \textbf{downto} $1$}
    \FOR{each $z$ in $L(i)$}
      \IF{$v_i$ is a leaf in $T$}
        \STATE $\textnormal{N}_z = 1$
      \ELSE
        \STATE $\displaystyle \textnormal{N}_z = \prod_{v \in A_i}
        \bigg( \sum_{y \in \text{adj}_z : \colorf(y) = \colorf(v)}
        \!\!\! \textnormal{N}_y \bigg)$  
      \ENDIF
    \ENDFOR
  \ENDFOR

  \RETURN $\sum_{z \in L(1)} \textnormal{N}_z$

\end{algorithmic}

\end{algorithm}

Since vertices in $V_G$ are processed according to their colors given
by Step~\ref{alg-nOfOccur:st1}, the values $\text{N}_y$, for all $y$
in $\text{adj}_z$, are previously computed before the computation of
$\text{N}_z$. This implies, using sum and product rules, that
$\text{N}_z$ is correctly computed.

The last step of Algorithm~$\eta(G,T)$ takes into account that each
occurrence of $T$ in $G$ has a vertex of color $1$, there are
$|L(1)|$ vertices in $G$ with color $1$, and, for each vertex $z$
with color $1$, $\text{N}_z$ is the number of occurrences of $T_{v_1}
\cong T$ in $G$ with vertex $z$.
Therefore, the number of occurrences of $T$ in $G$ is given by the sum of
all $\text{N}_z$ such that $z \in L(1)$, according to Step~8. 

\begin{theorem} \label{theo:numocc}
Let $G$ be a vertex-colored graph and $T$ a colorful tree such that
$|V_T|~\leq~|V_G|$. The number $\eta(G, T)$ can be computed in
$O(V_G+E_G)$ time.
\end{theorem}

\begin{proof}
The list of vertices to be processed can be established in a
preprocessing step in $O(V_T)$ time. Taking into account the computation
of $\text{N}_z$ for all vertices in $V_G$, Algorithm~$\eta(G, T)$
performs at most $|E_G|$ products and $|E_G|$ sums,
spending thus $O(E_G)$ time. Then, we compute $\sum_{z \in L(1)}
\text{N}_z$ in $O(V_G)$ time. Therefore, the running time for
calculating $\eta(G, T)$ is $O(V_G+E_G)$.
\end{proof}

Note that, with some preprocessing and specially for small motifs,
Algorithm~\ref{alg:numocc} can also be used when one is interested
only in colors, not in topology, i.e., ignoring structural
similarity. In order to do so, given a set of colors, we generate and
count all colorful trees containing exactly those colors, summing
the output values.

\section{\label{sec:infe} Inference}

\newcommand{\Prob}[1]{\protect{\ensuremath{\textup{Pr}\{#1\}}}}
\newcommand{\Ind}[1]{\protect{\ensuremath{\textup{I}\{#1\}}}}
\newcommand{\Esp}{\protect{\ensuremath{\textup{E}}}}
\newcommand{\Var}{\protect{\ensuremath{\textup{Var}}}}

In this section, we consider the problem of detecting subgraphs in a
network $G$ occurring with high frequency. The concept ``high
frequency'' here depends not only on the number of occurrences of a
given subgraph but also on the number of occurrences that the subgraph
appears in a random graph. In previous work, Schbath et
al.~\cite{SLS2009} define a null model, based on
Erdös~\cite{Erdos1947}, assuming that edges are independent and
distributed according to a Bernoulli distribution with the same
parameter $p \in (0, 1]$ for each edge. In addition, we add vertex
  colors to the random graph model and, since we are also considering
  topological structure to define motifs, our null model uses a
  parameter function that depends on the colors of vertices that are
  adjacent in the graph. 

\subsection{\label{sec:eval} Evaluating high frequency}

We propose a null model as a generalization of the Erdös-Rényi model,
where the occurrence of colors is considered. This approach takes into
account that some pairs of colors in vertex-colored graphs
representing biological networks are more likely to be connected than
others.

Let $a$ and $b$ be colors in $\colorset$ and let $G$ be a
vertex-colored graph. Let $|a| := |\{v \in V_G : \colorf(v) = a\}|$
and $|(a, b)| := |\{ (u, v) \in E_G : \colorf(u) = a \text{ and }
\colorf(v) = b\}|$. We define a parameter function (which implicitly
depends on $G$)
\[
p(a, b) = \frac{|(a, b)|}{|a| \cdot |b|}\:,
\]
for each pair of colors $a, b \in \colorset$. Notice that, for a
colorful graph, $|a| = |b| = 1$ and $|(a, b)| \in \{ 0, 1\}$, which
implies that $p(a, b)$ is either $0$ or $1$.

A random vertex-colored graph $R$ defined from a vertex-colored graph
$G$ is such that $V_R = V_G$, with same color function and each edge
$(x, y)$ occurs in $E_R$ independently with probability $0 \leq
p(\colorf(x), \colorf(y)) \leq 1$.  Thus, the probability of a
particular graph $R$ in a sample space $\mathcal{G} = \{
\text{vertex-colored graph }R : V_R = V_G\}$ is
\[
\Prob{R} = \!\! \prod_{(x,y) \in E_R} \!\! p(\colorf(x), \colorf(y)) \cdot
\!\! \prod_{(x,y) \not\in E_R} \!\! (1 - p(\colorf(x), \colorf(y)))\:.
\]

A \emph{candidate for occurrence of a motif $T$ in $R$ (or $G$)}, or
simply \emph{candidate}, is a vertex set $S \subseteq V_R$ such that
$\colorf(S) = \colorf(V_T)$ and $|S| = |V_T|$. We denote by
$\mathcal{S}$ the set of candidates of $T$ in $R$. We have the
following collection of straightforward results on candidates for
occurrence of a motif in a random graph.

\begin{lemma} \label{lem:prob+exp}
  Let $R$ be a random vertex-colored graph, $T$ a colorful tree, and
  $\mathcal{S}$ a set of candidates for occurrence of $T$ in $R$. For
  a candidate $S \in \mathcal{S}$, define the indicator random
  variable $X_S$ which states whether a candidate $S$ is an occurrence
  of $T$ and let $\mu$ be the expected value of $X_S$ for any
  candidate $S \in \mathcal{S}$. Let $X$ be a random variable
  representing the number of occurrences of $T$ in $R$. Then we have
  the following results:
  \begin{enumerate}
    \item[(i)] The number of candidates for occurrence of $T$ in $R$ is
      $
      |\mathcal{S}| = \prod_{t \in \colorf(T)} |t|\:,
      $
      and it can be obtained in $O(V_G)$ time.
    \item[(ii)] The expectation
      $
      \Esp[X_S] = \prod_{(u, v) \in E_T} p(\colorf(u), \colorf(v))
      $
      can be obtained in $O(V_T)$ time.
    \item[(iii)] The expected number of occurrences of $T$ in
      $R$ is given by
      $
      \Esp[X] = |\mathcal{S}| \cdot \mu\:,
      $
      and it can be obtained in $O(V_G)$ time.
    \item[(iv)] The variance of the number of occurrences of $T$ in
      $R$ is given by $ \Var[X] = \sum_{I \subseteq \colorf(V_T)} m(I)
      \cdot \mu^2/f(I)$, where the value $m(I)$ denotes the number of
      ordered pairs of candidates $S_i$ and $S_j$ in $V_R$ such that
      $\colorf(S_i \cap S_j) = I$, $f(I) = \prod_{(u, v) \in E_{T[I]}}
      \Prob{(\colorf(u), \colorf(v))}$, and $T[I]$ is the subgraph of
      $T$ induced by vertices whose colors are in $I$. Moreover,
      $\Var[X]$ can be obtained in $O(2^{V_T} \cdot V_T)$ time.
  \end{enumerate}
\end{lemma}

\begin{proof}
  \noindent (Proof for (\emph{i})). 
  For a vertex $t$ in $T$, recall that $|t|$ is the number of vertices
  in $V_G$ $(= V_R)$ with color $t$. Since $T$ is a colorful graph,
  the number of candidates of $T$ in $R$ is obtained directly by the
  product of the multiplicity of each color in $V_G$ and therefore the
  assertion follows. Moreover, that number can be computed in $O(V_G)$
  time.

  \medskip

  \noindent (Proof for (\emph{ii})).
  For a candidate $S$ in $\mathcal{S}$, we introduce an indicator
  random variable
  \[
  X_S = \left\{
  \begin{array}{ll}
    1\:, & \text{if candidate $S$ is an occurrence of $T$}\:,\\
    0\:, & \text{otherwise}\:.
  \end{array}
  \right.
  \]
  
  Notice that a necessary and sufficient condition for a candidate $S$
  being an occurrence of a colorful tree $T$ in $R$ is that vertices
  of $S$ with color $\colorf(u)$ and $\colorf(v)$ are adjacent in $R$,
  for each $(u, v) \in E_T$. Thus, the probability of $S$ being an
  occurrence of $T$ in $R$ is
  \begin{equation} \label{eq:probX_S}
    \Prob{X_S = 1} = \prod_{(u, v) \in E_T} p(\colorf(u),
    \colorf(v))\:.
  \end{equation}

  Since $X_S$ is an indicator random variable, we have
  \[
  \Esp[X_S] = 
  1 \cdot \Prob{X_S = 1} + 0 \cdot
  \Prob{X_S = 0}\:,
  \]
  and thus, by equality~(\ref{eq:probX_S}), the assertion
  follows. Notice also that $|E_T| = O(V_T)$ and therefore $\Esp[X_S]$
  can easily be computed in $O(V_T)$ time.

  \medskip

  \noindent (Proof for (\emph{iii})). 
  Notice from (\emph{ii}) that $\Esp[X_S]$ does not depend
  on the candidate $S$. That is, if $S$ and $S'$ are two candidates of
  $T$ in $R$, then $\Esp[X_S] = \Esp[X_{S'}]$. Thus, for simplicity,
  we denote $\Esp[X_S]$ by $\mu$, for any $X_S$ and $S$.

  Consider another random variable $X$ denoting the number of
  occurrences of a motif $T$ in $R$. That is, $X = \sum_{S \in \mathcal{S}}
  X_S$.

  It is immediate that
  \[
  \Esp[X] = \Esp\bigg[\sum_{S \in \mathcal{S}} X_S\bigg] = \sum_{S \in
    \mathcal{S}} \Esp[X_S] = \sum_{S \in \mathcal{S}} \mu =
  |\mathcal{S}| \cdot \mu\:,
  \]
  where the second equality follows from linearity of expectation. By
  (\emph{i}), $|\mathcal{S}|$ can easily be computed in $O(V_G)$ time
  and, by (\emph{ii}), $\mu$ can be computed in $O(V_T)$ time. It
  follows that $|\mathcal{S}| \cdot \mu$ can be computed in $O(V_G) +
  O(V_T) = O(V_G)$ time.

  \medskip

  \noindent (Proof for (\emph{iv})). 
  Recall that
  \begin{equation}
  \Var[X] = \Esp[X^{2}] - \Esp^{2}[X] = \Esp[X^{2}] -
  (|\mathcal{S}|\mu)^2\:, \label{eq:calc-var}
  \end{equation}
  where the last equality is valid from (\emph{iii}). Notice that
  \begin{equation} \label{eq:X^2}
    X^2 = \bigg( \sum_{S \in \mathcal{S}} X_S \bigg)^2 = \sum_{S_i, S_j
      \in \mathcal{S}} X_{S_i} X_{S_j}\:. 
  \end{equation}
  Since
  \[
  X_{S_i} X_{S_j} = \left\{ \begin{array}{ll}
    1\:, & \text{if } X_{S_i} = 1 \text{ and } X_{S_j} = 1\:, \\
    0\:, & \text{otherwise}\:,
  \end{array} \right.
  \]
  we have
  
  \begin{align} 
    \Esp[X_{S_i}\!X_{S_j}] &\!=\!\Prob{\!X_{S_i}\!X_{S_j}\!=\!1\!} 
    \!=\!\Prob{\!X_{S_i}\!=\!1 \text{ and } X_{S_j}\!=\!1\!}. \label{eq:Esp}
  \end{align}
  
  \noindent In addition, considering a set of colors $I$, we define
  \[
  \Prob{I} = \sum_{\substack{S_{i}, S_{j} \in \mathcal{S} :
      \\ \colorf(S_i \cap S_j) = I}} \Prob{X_{S_i} = 1 \text{
      and } X_{S_j} = 1}\:.
  \]
  It follows, using~(\ref{eq:X^2}),~(\ref{eq:Esp}) and 
  linearity of expectation, that
  
  \begin{align}
    \Esp[X^2] &= \Esp\bigg[\sum_{S_i, S_j \in \mathcal{S}} X_{S_i}
      X_{S_j} \bigg] = \sum_{S_i, S_j \in \mathcal{S}} \Esp[X_{S_i}
      X_{S_j}] \notag \\
    &=\!\!\sum_{S_i, S_j \in \mathcal{S}}\!\!\!\!\Prob{X_{S_i}\!=\!1 \text{ and }
      X_{S_j}\!=\!1} 
    \!=\!\!\!\!\sum_{I \subseteq \colorf(V_T)} \!\!\!\! \Prob{I}\:. \label{eq:EspX^2}  
  \end{align}
  
  Therefore, since $\Var[X] = \Esp[X^2] - (|\mathcal{S}|\mu)^2$ and we
  know how to compute $|\mathcal{S}|\mu$, if we are able to compute
  $\Prob{I}$ for each subset $I$ of $\colorf(V_T)$, then we can
  compute $\Esp[X^2]$ and $\Var[X]$.

  Now, consider $S_i, S_j \in \mathcal{S}$ such that $\colorf(S_i
  \cap S_j) = I$ and denote by $T[I]$ the subgraph of $T$ induced by
  vertices whose colors are in $I$. Let $f(I) := \prod_{(u,v) \in
    E_{T[I]}} \Prob{(\colorf(u), \colorf(v))}$. Since graphs with
  vertex subsets $S_i$ and $S_j$ share only the edges whose
  extremities are in $I$, we have that
  
  \begin{align*}
    \Prob{X_{S_i} = 1 \text{ and } X_{S_j} = 1} &= \frac{\Prob{X_{S_i} = 
        1} \cdot \Prob{X_{S_j} = 1}}{f(I)} \\
    &= \frac{\mu^2}{f(I)}\:.
  \end{align*}

  It follows that
  
  \begin{align*}
    \Prob{I} & = \sum_{\substack{S_i, S_j \in \mathcal{S} :
        \\ \colorf(S_i \cap S_j) = I}} \Prob{X_{S_i} = 1, X_{S_j} =
      1} \\ 
    & = \sum_{\substack{S_i, S_j \in \mathcal{S} :
        \\ \colorf(S_i \cap S_j) = I}} \frac{\mu^2}{f(I)} = m(I)
    \cdot \frac{\mu^2}{f(I)}\:,
  \end{align*}
  where $m(I)$ is the number of ordered pairs of candidates $S_i$ and
  $S_j$ in $V_R$ such that $\colorf(S_i \cap S_j) = I$. Therefore,
  the assertion follows.

  Denote the set $\colorf(V_T) - I$ by $\bar{I}$. There exist
  $\prod_{i \in I} |i|$ ways of choosing a subset of $|I|$ vertices
  that are in $S_i \cap S_j$, and there exist $\prod_{i \in \bar{I}}
  |i| \cdot \prod_{i \in \bar{I}} (|i|-1)$ ways of choosing subsets to
  play the role of $S_i - (S_i \cap S_j)$ and $S_j - (S_i \cap
  S_j)$. Therefore it follows that $m(I) = \prod_{i \in I} |i| \cdot
  \prod_{i \in \bar{I}} |i| \cdot \left( \prod_{i \in \bar{I}} (|i|-1)
  \right)$.

  Observe that $\mu$, $f(I)$ and $m(I)$ can be computed in $O(V_T)$
  time for each $I$, and thus we spend $O(V_T)$ time for calculating
  $\Prob{I}$ for each $I \subseteq \mathcal{S}$. Therefore, since
  there exist $2^{|V_T|}$ subsets of $V_T$, we have that $\Esp[X^2]$ can
  be computed in $2^{|V_T|} \cdot O(V_T) = O(2^{V_T} \cdot V_T)$ time,
  and this is the running time to compute $\Var[X]$.
\end{proof}

Recall that $\eta(G,T)$ denotes the number of occurrences of a
colorful tree $T$ in a graph $G$. Using the previous results, we can
state the chance of a random graph having a certain number of
occurrences of $T$.

\begin{theorem} \label{theo:number}
Let $G$ be a vertex-colored graph, $T$ a colorful tree and $X$ 
a random variable describing the number of occurrences
of $T$ in $G$. If $\eta(G, T) > \Esp[X]$, then
\begin{equation} \label{equ:Cheb}
\Prob{X \geq \eta(G, T)} \leq \frac{\Var[X]}{(\eta(G, T) - \Esp[X])^2}\:.
\end{equation}
\end{theorem}
\begin{proof}
After computing $\eta(G, T)$, and $\Esp[X]$ and $\Var[X]$ according to
Lemma~\ref{lem:prob+exp}(\emph{iii})
and Lemma~\ref{lem:prob+exp}(\emph{iv}), we use Chebyshev's inequality and
the result follows immediately. Such an inequality is used for
evaluating an upper bound for the chance of a random graph having at
least $\eta(G, T)$ occurrences of $T$ in $G$, suggesting that $T$ is
an actual motif in $G$ if $\Var[X] / (\eta(G, T) - \Esp[X])^2$ is
small.
\end{proof}

Moreover, we define a quality measure indicating how expected is the
number $\eta(G, T)$. For a real number $y$, we say that a colorful
tree $T$ is a \emph{$y$-motif} of a vertex-colored graph $G$ if
\[
y = \max \left\{ 0, 1 - \frac{\Var[X]}{(\eta(G, T) - \Esp[X])^2}
\right\}\:.
\]
Observe that the running time to compute $y$ is exponential, as
Lemma~\ref{lem:prob+exp}(\emph{iv}) states. However, we are usually
interested in small values of $|V_T|$. Typically for these cases we
can obtain $y$ very quickly as we will see in Section~\ref{sec:expe}.

\subsection{\label{sec:inmo} Inferring motifs}

Given the previous theoretical framework, we present now an algorithm
for inferring statistically significant colorful motifs in a given
vertex-colored graph. This algorithm is an enhanced version of that
in~\cite{RAS2015}, improved with significant speedup.

First, consider a vertex-colored graph $G$ representing a biological
network. A \emph{consistent colorful tree} with respect to $G$ is a
colorful tree $T$ such that $(u,v) \in E_T$ only if there exists
$(u',v') \in E_G$ such that $\colorf(u) = \colorf(u')$ and $\colorf(v)
= \colorf(v')$, for each pair of vertices $u, v$ in $T$. Intuitively,
a consistent colorful tree does not contain a pair of adjacent
vertices whose colors never occur on adjacent vertices in the given
network. Considering only consistent colorful trees when inferring
motifs is a (constant-factor) optimization, since in biological
networks usually some pairs of colors are never adjacent.

Since the direct inference for large motifs is highly time expensive,
the algorithm begins inferring motifs of some smaller initial size
$s$, using the results as a base for increasing the size of motifs
inferred up to a large size $g$. This is performed in two major
steps. Given a set of colors of interest $C \subseteq \colorset$
(e.g. colors representing a set of relevant reactions that one wants
to investigate in a metabolic network), a graph $G$, an initial size
$s$, a goal size $g$, a threshold $t$, and a score $y$, with $0 \leq y
\leq 1$, the first main step generates all possible subsets $\Gamma
\subseteq C$, with $|\Gamma| = s$. Then, for each $\Gamma$, it obtains
every consistent colorful tree $T$ of size $s$ such that $\colorf(T) =
\Gamma$, producing a set $\mathcal{T}$ of trees that occur at least
$t$ times in $G$ and which are $y$-motifs. In the second step, it
obtains several consistent colorful trees $T'$ from every $T \in
\mathcal{T}$ by adding a new vertex to $T$, searching $T'$ in $G$, and
obtaining a new set $\mathcal{T}'$ containing $y$-motifs that occur at
least $t$ times in $G$. At the end of this step, parameter $s$ is
incremented and the second step is repeated for $\mathcal{T}'$ as
$\mathcal{T}$. Algorithm~\ref{alg:inference}
(\textsc{Motif-Inference}) presents a pseudocode implementing these
steps.

\begin{algorithm}
\caption{\label{alg:inference} \textsc{Motif-Inference($C, G, s, g, t, y$)}}
\begin{algorithmic}[1]
  \footnotesize
  
  \REQUIRE{set of colors $C$, graph $G$, initial size
    $s$, goal size $g$, threshold $t$, score $y$ (with $0
    \leq y \leq 1$)}

  \ENSURE{a set $\mathcal{T}$ of trees occurring at least $t$ times in 
    $G$ and which are $y$-motifs} 

  \smallskip
  
  \STATE Let $\mathcal{P}(C)$ be the powerset of $C$ and
  $\mathcal{P}_s(C) \subseteq \mathcal{P}(C)$ such that $\Gamma \in
  \mathcal{P}_s(C)$ iff $\Gamma \in \mathcal{P}(C)$ and $|\Gamma| =
  s$ \label{mi:line:step1s}
  
  \STATE $\mathcal{T} \gets \emptyset$ \label{mi:line:step9999}

  \FOR{\textbf{each} $\Gamma$ in
    $\mathcal{P}_s(C)$} \label{mi:line:step1loop} 
  
    \FOR{\textbf{each} consistent colorful tree $T$ of size $s$
      such that $\colorf(V_T) = \Gamma$}
    
      \IF{$\eta(G, T) \geq t$ \textbf{and} $T$ is a $y$-motif}
    
        \STATE $\mathcal{T} \gets \mathcal{T} \cup
        \{T\}$ \label{mi:line:step1e}
    
      \ENDIF

    \ENDFOR
      
  \ENDFOR

  \WHILE{$s < g$} \label{mi:line:scatterbef1}
  
    \STATE $\mathcal{T}' \gets \emptyset$

    \FOR{\textbf{each} $T \in \mathcal{T}, u \in V_T$ \textbf{and} $a
      \in C\setminus \colorf(V_T)$} \label{mi:line:step2loop}
    
      \STATE $v \gets \text{a new vertex with color } a$

      \STATE $T':= T + \{uv\}$ 

      \IF{$T'$ is a consistent colorful tree, $\eta(G, T') \geq t$
        \textbf{and} $T'$ is a $y$-motif}
      
        \STATE $\mathcal{T}' \gets \mathcal{T}' \cup \{T'\}$

      \ENDIF

    \ENDFOR

    \STATE
    \textsc{Remove-Duplicates($\mathcal{T'}$)} \label{mi:line:rem_dupl}

    \STATE $\mathcal{T} \gets
    \mathcal{T}'$ \label{mi:line:scatterbef2}
    
    \STATE $s \gets s + 1$
    
  \ENDWHILE
    
  \RETURN $\mathcal{T}$%

\end{algorithmic}
\end{algorithm}
 
The duplicate removal in line~\ref{mi:line:rem_dupl} avoids repeated
motifs in $\mathcal{T}$, since they may arise in the second step of
the algorithm. For each $T \in \mathcal{T}$, a signature is generated,
which is unique for each colorful tree. Signatures are used together
with a hash table to detect and discard identical motifs.

Lines~\ref{mi:line:step1s} and~\ref{mi:line:step9999} of
\textsc{Motif-Inference} spend $O(s \cdot \binom{|C|}{s}) = O(s
|C|^{s})$ time. For finding occurrences of a motif in a graph we use
the Algorithm~\ref{alg:numocc}. This is a significant improvement over
the inference algorithm version in~\cite{RAS2015}, which used
Algorithm~\ref{alg:all-colorful}, since finding the number of
occurrences of some motif in $G$ is much faster than finding the
occurrences themselves. Given a graph $G$ with $|V_G|$ vertices and
$|E_G|$ edges and a colorful tree $T$, the number of occurrences of
$T$ in $G$ is given by $\eta(G, T)$ in $O(V_G+E_G)$ time. Thus, the
running time of lines 3-6 is $O\left(|C|^s \cdot s^{s-2} \cdot
(V_G+E_G) \right)$, since $\binom{|C|}{s} = O(|C|^s)$ and, from
Cayley's formula, the number of colorful trees with $s$ distinct
colors is~$s^{s-2}$. Besides that, the algorithm spends $O(|C|^g
g^{g-2} \cdot (V_G+E_G))$ time in the loop of lines 7-16 in the worst
case. Therefore, the running time for \textsc{Motif-Inference} is
bounded by lines 7-16, since $g > s \geq 1$.

We note that the inference algorithm returns a set $\mathcal{T}$ of
$y$-motifs in $G$. After that, the \textsc{All-Colorful} algorithm
(Algorithm~\ref{alg:all-colorful}) can be used to find all occurrences
in $G$ of some $T \in \mathcal{T}$ of interest.

We also observe that only few changes in the algorithm are needed for
building a parallel version. First, it must divide $\mathcal{P}_s(C)$
(line \ref{mi:line:step1s}) among $p$ processors, assigning a set
$\Gamma_i$ to each processor $i \in \{1, \ldots, p\}$. Thus, each
processor will finish the major first step of the motif inference
(lines 3-6) having a pruned set $\mathcal{T}_i$. Then, it must
rebalance the sets $\mathcal{T}_i$ among processors, e.g., by sending
all $\mathcal{T}_i$ to a master processor and spreading them equally
among processors or by some load balancing heuristic. After the loops
in lines 3-6 and lines 9-13, the algorithm gathers at the master
processor a set $\mathcal{T} = \mathcal{T}_1 \cup \dotsb \cup
\mathcal{T}_p$, scattering it again before
lines~\ref{mi:line:scatterbef1} and~\ref{mi:line:scatterbef2}.

\section{\label{sec:expe} Experimental results}

We implemented Algorithms~\ref{alg:MCG} to~\ref{alg:inference} and a
parallel version of Algorithm \ref{alg:inference}, comparing their
results to publicly available tools. We also implemented some methods
described in this paper to evaluate them in practice
(e.g. Formula~(\ref{eq:calc-var}) to calculate approximately the
variance of the number of occurrences for some colorful
tree). Sequential experiments were made on a single Intel i7 3.40GHz
processor (one thread), while parallel experiments were run on a
cluster with 40 Intel Xeon Quad-core 2.4GHz machines, 2 threads per
core, totaling up to 320 simultaneous processes.

In the following paragraphs seven sets of experiments are
presented. The first set evaluates the variance calculated by
Formula~(\eqref{eq:calc-var}), and the second evaluates the method for
accessing motif exceptionality. Following, the performance of
Algorithm~\ref{alg:numocc} is measured for counting occurrences of
non-topological motifs as described above. Next, regarding the
inference of larger motifs, the quality of the incremental heuristic
of Algorithm \textsc{Motif-Inference} ($s < g$) is compared against
the inference directly for the goal size ($s = g$). The following set
of experiments correlate the number of motifs in the output and the
running time of the Algorithm \textsc{Motif-Inference} to some of its
parameters.  Algorithm~\ref{alg:inference} is compared then to the
tool \textsf{MOTUS} (inference), and finally
Algorithm~\ref{alg:all-colorful} is compared to the tool
\textsf{Torque} (search).

In all the experiments except for the last we used the meta\-bol\-ic
network of \textit{E.~coli} strain K-12 from
BioCyc~\cite{Caspi-etal-2016}. We removed big molecule reactions,
pathway holes and some highly connected compounds, such as
H\textsubscript{2}O, H\textsuperscript{+} and CO\textsubscript{2},
resulting in a reaction graph with 1576 vertices and 3657
edges. Vertices representing reactions associated to enzymes whose EC
numbers coincide in the first three positions were colored with the
same color. We selected as colors of interest the 30 most frequent
colors, since high frequency colors should challenge our algorithms in
workload and memory consumption most. All remaining vertices, about
half of the total number, were colored by the same ``dummy'' color.

We address the accuracy of the sample variance~\eqref{eq:sample-var}
for a colorful tree compared to the calculated
variance~\eqref{eq:calc-var} (see Fig~\ref{fig:samplevar}),
corroborating the correctness of~\eqref{eq:calc-var}. The sample
variance was obtained generating random graphs based on the real
network. They have the same vertices as the base graph, and an edge is
added connecting vertices $u$ and $v$ according to the frequency of
edges connecting vertices with colors $\colorf(u)$ and $\colorf(v)$ in
the base graph.  This model captures the fact that some reactions and
their context are more frequent and/or better conserved during
evolution than others~\cite{DPS07,PSB05}. We have made experiments
with motifs of size 6 and 9, observing that the sample variance always
converged to the calculated variance after generating a large amount
of random graphs ($\approx 5 \times 10^{6}$). Moreover, the calculated
variance is obtained in a fraction of a second whereas the sample
variance takes from minutes to hours to be calculated.

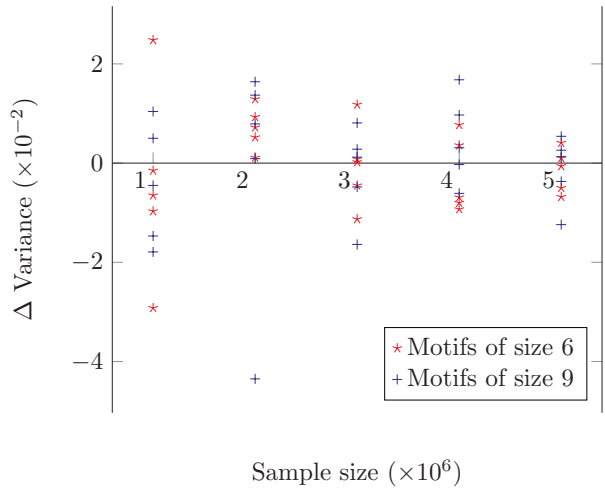
\begin{figure}[h]%
\centering%
\begin{tikzpicture}[scale=0.95]
\begin{axis}[
    align=center,
    xlabel={Sample size ($\times 10^6$)},
    ylabel={$\Delta$ Variance ($\times 10^{-2}$)},
    x label style={at={(axis description cs:0.5,0.0)},anchor=north},
    axis x line*=middle,
    enlargelimits=true,
    legend pos=south east,
    ytick scale label code/.code={},
    xtick scale label code/.code={},
    xticklabel style = {xshift=-.5em},
    ]
\addplot[
scatter/classes={
        6={mark=star,red},%
        9={mark=+,blue}%
},
scatter, only marks,
scatter src=explicit symbolic]
table[x=samplesize,y=dst,meta=motifsize] {sample-var.dat};
\legend{{Motifs of size 6}, {Motifs of size 9},
        }
\end{axis}
\end{tikzpicture}
\caption{\label{fig:samplevar} Distance from sample variance to
  calculated variance by sample size, for motifs of size 6 and 9.}
\end{figure}

For comparison against the null model, we ran the motif inference
algorithm for the \textit{E.~coli} meta\-bol\-ic network, using
different combinations of parameters, and for random graphs generated
as described above. As can be seen in Fig~\ref{fig:nullmodel}, for
$y=0.999$ as score for $y$-motifs, the cuts made by the algorithm on
random graphs drop significantly the amount of motifs found when they
grow in size, which is the opposite behavior of inference for real
networks. The same occurs for $y = 0.99$, in different scale (not
shown). This indicates that the method proposed in this work is able
to evaluate the exceptionality of motifs, i.e., the occurrence of a
motif more frequently than expected at random. When using the method
for inferring relevant motifs in practice, however, a higher $y$ value
or some other filter must be used, resulting in output sets of
reasonable sizes.

Regarding counting occurrences ignoring structural similarity as
mentioned above, we performed experiments for this together with
Algorithm~\ref{alg:inference}. Given sets of 6, 7, 8, 9 and 10 colors
of high frequency, we inferred motifs with these colors ignoring
topology and the process took 0.01, 0.15, 0.3, 0.4 and 0.7 seconds,
respectively, on a single machine.

We performed tests regarding differences between two methods:
incremental inference with $s = 5$, $g = 7$ and straight generation of
all motifs of size $g \in \{6,7\}$. As we can see in
Fig~\ref{fig:incr-straight}, there is little difference in the amount
of motifs found, particularly when setting $y$ to reasonable
values. Apart from being much faster, the incremental method misses
few motifs compared to generating all possibilities. Both tests were
performed in a parallel environment, the first (incremental) method
taking 8 seconds on average and the second about 208 seconds for $g =
7$.

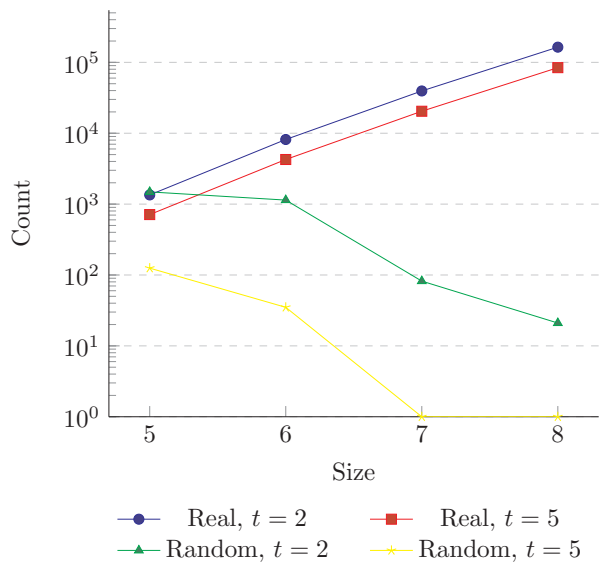
\begin{figure}[h]%
\centering%
\begin{tikzpicture}[scale=0.95]
\begin{semilogyaxis}[ 
	axis x line*=bottom, 
	axis y line*=left,
    align=center,
    xlabel={Size},
    ylabel={Count},
    ymin=1,
    ymajorgrids=true,
    grid style=dashed,
    xtick={5,6,7,8},
    legend style={ draw=none,
                   at={(0.0,-0.20)},
                   anchor=north west,
                   legend columns=2,
                   /tikz/every even column/.append style={column sep=0.5cm},
                 },
    ]
\addplot+[discard if not={type_threshold}{real_2}]
table[x=size,y=motifs] {nullmodel-0.999.dat};
\addplot+[discard if not={type_threshold}{real_5}]
table[x=size,y=motifs] {nullmodel-0.999.dat};
\addplot+[discard if not={type_threshold}{random_2},mark=triangle*,green,mark options={fill=green!60!black}]
table[x=size,y=motifs] {nullmodel-0.999.dat};
\addplot+[discard if not={type_threshold}{random_5},yellow]
table[x=size,y=motifs] {nullmodel-0.999.dat};
\legend{{Real, $t=2$},{Real, $t=5$},{Random, $t=2$},{Random, $t=5$}}
\end{semilogyaxis}
\end{tikzpicture}
\caption[]{\label{fig:nullmodel} Motif inference for real (\textit{E.~coli}) and random networks (null model); motifs are kept only if $y \ge 0.999$ and they occur more than the threshold $t$. Other parameters are $s = 5$, $g = 8$ and $C=\{1, \dotsc, 30\}$ (we infer motifs considering the 30 higher frequency colors.)}
\end{figure}

\begin{figure}[h]%
\centering%
\begin{tikzpicture}[scale=0.95]
\begin{axis}[ 
	axis x line*=bottom, 
	axis y line*=left,
    align=center,
    title={Motif inference comparison for some $y$ values,\\ $C=\{1, \dotsc, 30\}, s = 5, g = 7, t = 10$},
    xlabel={Size},
    ylabel={Count ($\times 10^4$)},
    ymin=1,
    ymajorgrids=true,
    grid style=dashed,
    xtick={5,6,7},
    ytick scale label code/.code={},
    legend style={ draw=none,
                   at={(0.10,-0.20)},
                   anchor=north west,
                   legend columns=2,
                   /tikz/every even column/.append style={column sep=0.5cm}
                 },
    ]
\addplot+[discard if not={y-motif}{0.1},red,mark options={fill=red!80!black},mark=square*]
table[x=size,y=count] {incr-straight.dat};
\addplot+[blue,mark options={fill=blue!80!black},mark=*]
coordinates {(6,4852)(7,19878)};
\addplot+[discard if not={y-motif}{0.99},mark=triangle*,green,mark options={fill=green!60!black}]
table[x=size,y=count] {incr-straight.dat};
\addplot+[discard if not={y-motif}{0.999},yellow]
table[x=size,y=count] {incr-straight.dat};
\legend{{$y=0.1$},{Straight},{$y=0.99$},{$y=0.999$}}
\end{axis}
\end{tikzpicture}
\caption[]{\label{fig:incr-straight} Motif counting for incremental
  method (with $s = 5$ and $g = 7$) compared to non-in\-cre\-men\-tal
  (directly for sizes 6 and 7). We infer motifs considering the 30
  higher frequency colors and set $t = 10$.}
\end{figure}
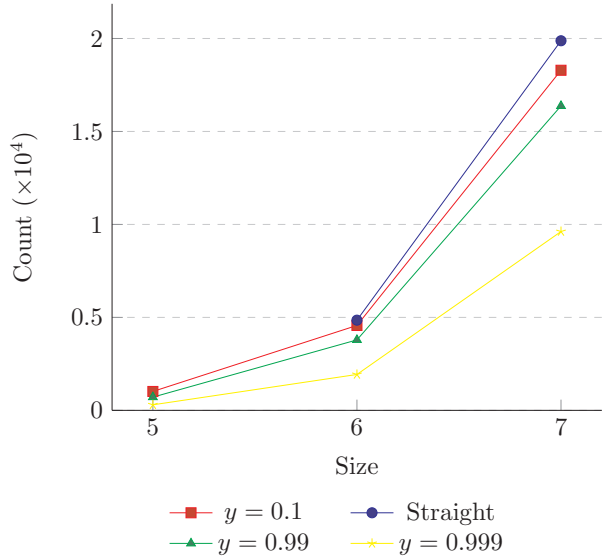

Note that in the experiments for motif inference, the number of found
motifs can grow exponentially.  Even when starting with a small set of
motifs, the incremental steps may lead to a large result set. In the
parallel environment, for motifs of size up to 12, the algorithm took
250 seconds for $y = 0.9999$ and $t = 30$, and 1137 seconds for $y =
0.999$ and $t = 40$, finding 167,673 and 262,819 topological motifs,
respectively. Ignoring topology in the output, i.e., grouping results
by color, the most frequent set of colors occurred more than $10^{5}$
times for both $y$ values. Importantly, when limiting the number of
motifs kept in memory to a few thousand, we managed to infer motifs up
to size 17 in about 160 seconds in parallel.

We compared Algorithm~\ref{alg:inference} with
\textsf{MOTUS}~\cite{LFS2005,LFS2006} in inference mode, using default
parameters (Table~\ref{tab:comp-motus}). While our algorithm
calculates $y$-motif scores and discards those below a given $y$
value, the only way of accessing motif exceptionality in
\textsf{MOTUS} is generating random networks, by default 100,
calculating and displaying a p-value. We grouped results output by
Algorithm~\ref{alg:inference} by color sets, allowing the analysis of
motifs disregarding their topologies, additionally displaying results
sorted by $y$-motif score.  Our incremental method is fast even for
greater sizes (see Table~\ref{tab:comp-motus}), while maintaining the
quality of results close to those obtained when generating and
counting all possibilities (non-incremental), as shown in
Fig~\ref{fig:incr-straight}. Furthermore, the speedup is significant
for tests performed in parallel.  We observe that \textsf{MOTUS} is
considerably slower, which is reasonable, since the set of colored
motifs is a superset of colorful motifs. However, to the best of our
knowledge there is no other tool available for colorful motif
inference which may be used for comparison.

\begin{table}[ht]\setlength{\tabcolsep}{1.8em}\setlength\extrarowheight{1pt}\centering
\caption{\label{tab:comp-motus} Running time for motif inference ($s = 4$, $y = 0.999$).}%
\begin{tabular}{crrr}
\hline
\multirow{2}{*}{Motif size} & \multicolumn{1}{c}{\multirow{2}{*}{\textsf{MOTUS}}} &
\multicolumn{2}{c}{Algorithm~\ref{alg:inference}}\\ \cline{3-4} 
           &                    & sequential  & parallel \\ \hline
4          & 200 s   & 3.9 s        & 2.1 s     \\
5          & 2100 s & 4.0 s        & 2.3 s     \\
6          & 22000 s & 4.8 s        & 3.2 s     \\
7          & 220000 s & 13.4 s       & 4.1 s     \\ \hline
\end{tabular}
\end{table}

We also ran experiments described in~\cite{Bruckner-etal-2010} to
compare Algorithm~\ref{alg:all-colorful} (\textsc{All-Colorful}) to
\textsf{Torque} regarding searching known protein complexes, counting
how many of them occur at least once in a PPI network (a
\emph{match}).  To evaluate the quality of the matches found we used
the \emph{Functional Coherence} method with the same parameters as
in~\cite{Bruckner-etal-2010}. A set of proteins (match) found to be
functionally coherent by the
\href{http://go.princeton.edu/cgi-bin/GOTermFinder}{\sf GO TermFinder}
tool~\cite{Boy04} with respect to the Gene Ontology (GO)
annotation~\cite{GO00} has a good quality and is likely to be used in
prediction of GO annotations for these proteins, when not available.
The network data was obtained from the
\href{http://www.cs.tau.ac.il/~Ebnet/TORQUE_Input_format.htm}{\sf
  Torque} website, for which we had to add some missing protein
sequences, while query motifs used on its experiments were found in
the \href{http://igm.univ-mlv.fr/AlgoB/gramofone/}{\sf GraMoFoNe}
website~\cite{Blin-etal-2010}. The threshold was set to $10^{-7}$ for
BLAST e-values and to $0.0$ for protein interaction probabilities,
meaning we connect pairs of vertices representing proteins with
interaction probability greater than zero.  The queries were processed
by our algorithm in parallel, and about 97\% of them finished in no
more than 3 seconds.  Table~\ref{tab:comp-torque} presents the total
number of \textsf{Torque} matches (novel and previously known)
reported in their paper compared with the counts found by the
Algorithm~\ref{alg:all-colorful}. The table also presents how many
matches found by the latter are functionally coherent. Compared
to~\cite{Bruckner-etal-2010}, the number of matches found to be
functionally coherent is very satisfactory.

\begin{table}[ht]\setlength{\tabcolsep}{1em}\setlength\extrarowheight{1pt}\centering
\caption{\label{tab:comp-torque} Protein complex (motif) search in PPI networks.}%
\begin{tabular}{cccccc}
\hline
\multirow{2}{*}{Network} & \multirow{2}{*}{Complex} &
\multicolumn{2}{c}{Matches} & \multicolumn{2}{c}{Functional Coherence} \\ \cline{3-4} \cline{5-6}
        &      	 & \textsf{Torque} & Alg.~\ref{alg:all-colorful} & \textsf{Torque} & Alg.~\ref{alg:all-colorful}
\\ \hline  
Yeast	& Bovine & 4 	  & 4	& 4  & 3  \\
        & Mouse	 & 18	  & 19  & 16 & 12 \\
        & Rat	 & 26	  & 22  & 19 & 8  \\
Fly	& Bovine & 0	  & 1   & 0  & 1  \\
	& Mouse	 & 13	  & 21  & 0  & 7  \\
	& Rat	 & 35	  & 35  & 17 & 8  \\
Human	& Bovine & 4	  & 7   & 2  & 5  \\
	& Mouse	 & 58	  & 113 & 32 & 66 \\
	& Rat	 & 49	  & 111 & 32 & 55 \\ \hline

\end{tabular}
\end{table}

\section{\label{sec:conc} Conclusion}

In this work we studied the search and inference of different
constraints of topological colored motifs in vertex-colored
graphs. Such studies could help us in the understanding of the
computational complexity for many related problems.

Considering a motif as a colorful tree $T$ and an occurrence as a
subset $S$ of vertices of a vertex-colored graph $G$ such that $|V_T|
= |S|$ and $T \cong G[S]$, we gave a simple linear time algorithm for
the searching problem. We also presented a method for enumerating all
occurrences of $T$ in $G$ and described a linear time algorithm to
compute the number of occurrences of $T$ in $G$. For the inference
problem, we developed a variant of the Erdös model, where we take into
account the colors of vertices as a parameter to calculate the
probability of a given colorful tree $T$ being a motif. The comparison
between the number of occurrences of a colorful tree $T$ in a given
subgraph and the expectation of the number of occurrences of $T$ in a
random graph obtained by the presented method can decide whether $T$
is a motif in $G$ or not, which takes time exponential in the size of
$V_T$, but is very fast in practice if $|V_T| \leq 20$.

When there exists a high rate of noise in a biological network,
especially due to missing/wrong data, we can allow finding occurrences
of motifs with gaps, meaning that some vertices can be part of an
occurrence of a motif in the network, not in the motif itself. Thus,
it is important to allow a flexible search of motifs when we cannot
find exact matches, only approximate. An approximate occurrence of a
motif is referred to as a \emph{gap} among some vertices of the given
network~\cite{LFS2006}. Regarding highly frequent motifs with gaps,
notice first that we can obtain the number of subgraphs of a given
vertex-colored graph that are isomorphic to a colorful motif with gaps
by making few changes to an algorithm given in~\cite{AS2013}. Then,
considering induced subgraphs, given a random vertex-colored graph $G$
and a colorful tree motif $T$ with gaps, a naive method to calculate
the expected value and variance of the number of occurrences is quite
slow. In this case, we can compute those measures using the following
strategy. Let $\mathcal{G}$ be the set of all random vertex-colored
graphs defined from $G$. The expected value and variance of the number
of occurrences of $T$ in $G$ are

\begin{align*}
  \Esp[X] &= \sum_{R \in \mathcal{G}} X(R) \cdot
  \Prob{R}\:, \hspace{0.8cm} \text{and} \\ 
  \Var[X] &= \sum_{R \in \mathcal{G}} (X(R) - \Esp[X])^2 \cdot
  \Prob{R}\:,  
\end{align*}

\noindent where $R$ is a random graph generated, according to the
previous section, and $X(R)$ and $\Prob{R}$ denote the number of
occurrences of $T$ in $R$ and the probability of $R$ to be generated,
respectively. Then, the main idea is generating explicitly all graphs
in $\mathcal{G}$ and counting the number of occurrences of each of
them. However, notice that $|\mathcal{G}|$ can be very large. Thus, we
generate randomly a smaller set $\mathbb{G}$ of random vertex-colored
graphs, representing a sample from $\mathcal{G}$, and then calculate

\begin{align}
  \Esp[X] &\simeq \frac{1}{|\mathbb{G}|} \sum_{R \in \mathbb{G}}
  X(R) \hspace{0.5cm} \text{and} \nonumber \\ 
  \Var[X] &\simeq \frac{1}{|\mathbb{G}|} \sum_{R \in \mathbb{G}} (X(R)
  - \Esp[X])^2\:. \label{eq:sample-var}
\end{align}

\noindent This approach can also be used when gaps are not considered
and $|V_T|$ is large, especially for calculating $\Var[X]$. However,
these ideas should be extended further.

Our concept of occurrence for topological motifs extends to weights of
occurrences, gaps in the networks, and deletions in motifs. The
results of such extensions can be used for proposing more general
problems, where the structure of the motifs is not considered.

\bibliographystyle{amsplain}
\bibliography{pre-print}

\end{document}